\documentclass[journal,a4paper]{IEEEtran}
\usepackage{algpseudocode}

\IEEEoverridecommandlockouts

\usepackage{cite}
\usepackage{lmodern,textcomp}
\usepackage{graphics}
\usepackage{multicol}
\usepackage{lipsum}
\usepackage{color}
\usepackage{amsthm}
\usepackage{amssymb}
\usepackage{mathtools}
\usepackage{amsbsy}
\usepackage{setspace}
\usepackage{float}
\usepackage{lettrine}
\usepackage{newclude}
\usepackage[normalem]{ulem}
\usepackage{latexsym}
\usepackage{multirow}
\usepackage{rotating}
\usepackage{booktabs}
\usepackage{wasysym}
\usepackage{mathrsfs}
\usepackage{bbm} 
\usepackage{epsfig,cite,amsfonts,psfrag}
\usepackage{apptools}
\usepackage{graphicx}
\usepackage[table]{xcolor}
\usepackage{chngcntr}
\usepackage{blindtext}

\usepackage{hyperref}
\def\BibTeX{{\rm B\kern-.05em{\sc i\kern-.025em b}\kern-.08em
    T\kern-.1667em\lower.7ex\hbox{E}\kern-.125emX}}
    
\usepackage{optidef}
\usepackage{accents}
\usepackage{dsfont}
\usepackage{siunitx}
\usepackage{overpic}
\usepackage{amsmath,amssymb,amsfonts}
\usepackage{comment}
\usepackage{graphicx}
\usepackage{textcomp}
\usepackage{xcolor}
\usepackage{optidef}

\usepackage{cite}
\usepackage{subcaption}
\usepackage{caption}

\usepackage{lipsum}
\usepackage{overpic}
\usepackage{amssymb}
\usepackage{amsmath}
\usepackage{array}
\usepackage{color}
\usepackage{url}
 
\usepackage{optidef}
\usepackage{times}
\usepackage{fancyhdr,graphicx}
\usepackage[ruled,linesnumbered,noend]{algorithm2e}

\usepackage{algorithmicx}
\usepackage{algcompatible}

\IEEEoverridecommandlockouts
\usepackage{cite}
\usepackage{graphicx}
\usepackage{url}
\usepackage{relsize}
\usepackage{hhline}
\usepackage{nicefrac}
\usepackage{amssymb,bm,upgreek}
\usepackage[table]{xcolor}
\usepackage{multirow}

\theoremstyle{plain}

\newtheorem{remark}{Remark}
\newtheorem{proposition}{Proposition}

\newcommand*{\LongState}[1]{\STATE
\parbox[t]{0.9\linewidth-\algorithmicindent-\algorithmicindent}{#1\strut}}

\newcommand{\vect}[1]{\boldsymbol{#1}}

\def\Ttran{\tiny \mbox{$\mathrm{T}$}}

\begin{document}

\title{Quantized Precoding for Maximizing Sum Rate in MU-MIMO Systems with Constrained Fronthaul}

\author{\IEEEauthorblockN{Yasaman Khorsandmanesh, \textit{Student Member, IEEE}, Alva Kosasih, \textit{
Member, IEEE}, Emil Björnson, \textit{Fellow, IEEE}, and Joakim Jaldén, \textit{Senior Member, IEEE}.
\thanks{Y. Khorsandmanesh and E. Björnson are with the Division of Communication Systems, KTH Royal Institute of Technology, Stockholm, Sweden (E-mails: \{yasamank, emilbjo\}@kth.se). A. Kosasih is with Nokia Standards, Espoo, Finland (E-mails: alva.kosasih@nokia.com). J. Jaldén is with the Division of Information Science and Engineering, KTH Royal Institute of Technology, Stockholm, Sweden (E-mail: jalden@kth.se).}
\thanks{This work was supported by the Knut and Alice Wallenberg Foundation through the WAF program.}
\thanks{A preliminary version of this work was presented at the IEEE International Conference on Communications (ICC), Rome, Italy, June 2023 \cite{khorsandmanesh2023fronthaul}.}
}}

\maketitle

\begin{abstract}  
This paper studies a downlink multi-user multiple-input multiple-output (MU-MIMO) system, where the precoding matrix is computed at a baseband unit (BBU) and then transmitted to the remote antenna array over a limited-capacity digital fronthaul. The limited bit resolution of the fronthaul introduces quantization effects that are explicitly modeled. We propose a novel sum rate maximization framework that directly incorporates the quantizer's constraints into the precoding design. The resulting maximization problem, a non-convex mixed-integer program, is addressed using a new iterative algorithm inspired by the weighted minimum mean square error (WMMSE) methodology. The precoding optimization subproblem is reformulated as an integer least-squares problem and solved using a novel sphere decoding (SD) algorithm.
Additionally, a low-complexity expectation propagation (EP)-based method is introduced to enable the practical implementation of quantized precoding in MU-massive MIMO (MU-mMIMO) systems. Furthermore, numerical evaluations demonstrate that the proposed precoding schemes outperform conventional approaches that optimize infinite-resolution precoding followed by element-wise quantization. We also propose a heuristic quantization-aware precoding method with comparable complexity to the baseline but superior performance. In particular, the EP-based approach offers near-optimal performance with substantial complexity reduction, making it well-suited for real-time MU-mMIMO applications. 
\end{abstract}

\begin{IEEEkeywords}
Sum rate maximization, constrained fronthaul, weighted minimum mean square error, quantization-aware precoding, sphere decoding, expectation propagation.
\end{IEEEkeywords}

\section{Introduction}

Multiple-input multiple-output (MIMO) systems enable high data rates through spatial multiplexing of multiple user equipments (UEs) on the same time-frequency resource \cite{Gesbert2007a}. Given a knowledge of channel state information (CSI), a base station (BS) equipped with multiple antennas can transmit simultaneously to several UEs with different directivities, such that desired signals can be pointed towards the UE while inter-UE interference can be suppressed. This is known as precoding design.
In precoding design, it is common to maximize the sum rate, or a weighted sum rate, under a constraint on the total transmit power \cite{bjornson2012robust}; however, sum rate maximization is known to be NP hard \cite{liu2010coordinated}. One popular practical approach for sum-rate maximization is the iterative weighted minimum mean square error (WMMSE) algorithm, which finds locally optimal solutions with affordable complexity \cite{shi2011iteratively}.

A 5G BS typically consists of two main components: an advanced antenna system (AAS) and a baseband unit (BBU). The AAS is a box containing the antenna elements and their respective radio units (RUs). The BBU performs the digital processing related to the received uplink data and transmitted downlink data. The AAS and BBU are connected through a digital fronthaul. The integration of antennas and radios into a single box has made multi-user MIMO (MU-MIMO) practically feasible \cite{Bjornson2019d} and enables the BBU to be virtualized in an edge cloud by migration to the centralized radio access network architecture \cite{peng2015fronthaul}. The new implementation bottleneck in these systems is the limited fronthaul capacity and the quantization errors it creates. The uplink/downlink data and combining/precoding coefficients are sent over this digital fronthaul and therefore must be quantized to a finite resolution. 
\subsection{Prior Work}
Downlink  MU-MIMO systems have been widely studied in previous literature regarding impairments in analog hardware \cite{Bjornson2014a} and the effect of low-resolution digital-to-analog converters (DACs) \cite{jacobsson2017quantized}. These prior works are characterized by distortion created either in the RU,  in the analog domain, or in the converters. Therefore, the transmitted signal is distorted after the precoding. The effect of limited fronthaul capacity is studied in \cite{parida2018downlink}, but the precoding design was not quantized. In \cite{khorsandmanesh2023optimized}, the authors proposed a fronthaul quantization-aware precoding design that minimizes the sum MSE, which will generally not maximize the sum rate. Many previous studies suggest designing a precoding matrix by maximizing the sum rate, often using the WMMSE approach; see \cite{christensen2008weighted,zhao2022rethinking} and references therein. Nevertheless, they consider ideal hardware or other types of distortion than precoding quantization.

Alternatively, symbol-level precoding techniques, where the transmitted signals are designed based on the knowledge of both CSI and the data symbols, have been recently proposed for downlink MU-MIMO systems with low-resolution DACs \cite{masouros2010correlation,tsinos2018symbol}. In \cite{tsinos2018symbol}, a novel symbol-level precoding technique is developed supporting systems with DACs of any resolution, and it is applicable for any signal constellation.  Unlike these non-linear precoding schemes, our linear precoding design is independent of the data and can be used for an arbitrary number of data symbols, leading to vastly lower complexity.

This interface between the BBU and AAS must carry received uplink signals to be decoded at the BBU and precoded downlink signals, which are computed at the BBU. In this paper, we analyze and mitigate the precoding distortion that occurs over this digital interface when the precoding matrix is quantized before the transmit signal is computed; that is, before the quantized precoding matrix is multiplied with the data symbols at the AAS.

We develop a benchmark method for computing the precoding matrix in an optimized manner and also propose a low-complexity alternative. The focus is on block-level precoding, where the same quantized precoding matrix is utilized for all the symbols in a coherence block. The SD offers an optimal solution for the developed problem formulation in this paper. However, it is associated with a high computational complexity, which can hinder its implementation. This paper utilizes the problem formulation from the conference version of our work \cite{khorsandmanesh2023fronthaul} and provides an alternative solution by leveraging a low complexity algorithm, i.e., an EP-based approach that did not appear in \cite{khorsandmanesh2023fronthaul}. The proposed algorithms can operate with any channel model, irrespective of how the CSI is obtained, through codebook-based feedback or uplink pilot signaling. Prior works, such as \cite{minka2013expectation,kschischang2001factor}, have demonstrated the efficacy of EP in various signal-processing tasks, highlighting its ability to achieve a favorable trade-off between complexity and performance. EP works by approximating the posterior distribution of the precoding matrix with a tractable Gaussian distribution that is iteratively updated. This approach has a significantly lower computational complexity than the SD approach, while retaining similar performance. Therefore, it is particularly attractive for large-scale systems with stringent fronthaul and processing constraints. 

\subsection{Contributions}
This paper focuses on MU-MIMO downlink transmission with limited-capacity fronthaul. We propose a transmit precoder that locally maximizes the sum rate under a power constraint, with the precoding matrix restricted to a discrete quantization codebook. The main contributions are:

\begin{itemize}
    \item We formulate the sum-rate maximization under fronthaul quantization as a \emph{quantization-aware precoding problem}, resulting in a non-convex mixed-integer optimization. We solve it using an iterative \emph{WMMSE-based} algorithm, where each step reduces to an integer least-squares problem solved via sphere decoding (SD).

    \item To overcome the prohibitive computational complexity of the SD-based method, we propose a low-complexity \emph{expectation propagation (EP)-based} algorithm that achieves near-SD performance with significantly lower complexity.
    \item We introduce a \emph{quantization-unaware} baseline and develop a low-complexity heuristic refinement algorithm that improves the sum rate with similar complexity.
    \item We analyze the  complexity and evaluate performance under Rician fading with  uniform linear arrays (ULA), and uniform planar arrays
(UPA) configurations, thereby characterizing the performance--complexity trade-offs.
\end{itemize}

\subsection{Notation and Preliminaries}
The sets of real and complex numbers are denoted by $\mathbb{R}$ and $\mathbb{C}$, respectively. Matrices and vectors are represented by bold uppercase and lowercase letters, e.g., $\boldsymbol{X}$ and $\boldsymbol{x}$. The $(m,k)$-th entry of $\boldsymbol{X}$ is $X_{m,k}$, its $k$-th column is $\boldsymbol{x}_k$, and the $m$-th element of $\boldsymbol{x}$ is $x_m$. The subvector $\boldsymbol{X}_{m,k:n}$ denotes the entries of the $m$-th row from column $k$ to $n$, and for $\boldsymbol{p} \in \mathbb{R}^N$, $\boldsymbol{p}_{k:n}$ denotes the subvector from index $k$ to $n$. The identity matrix of size $M$ is denoted by $\boldsymbol{I}_M$, while $\boldsymbol{1}_{M \times K}$ and $\boldsymbol{0}_{M \times K}$ denote all-one and all-zero matrices of size $M \times K$. The absolute value, Euclidean norm, and Frobenius norm are denoted by $|\cdot|$, $\|\cdot\|_2$, and $\|\cdot\|_{\mathrm{F}}$, respectively. The Kronecker product is written as $\otimes$. Matrix operations include transpose $(\cdot)^\mathrm{T}$, Hermitian transpose $(\cdot)^\mathrm{H}$, inverse $(\cdot)^{-1}$, vectorization $\mathrm{vec}(\cdot)$, and trace $\mathrm{tr}(\cdot)$. The real and imaginary parts are denoted by $\mathfrak{R}\{\cdot\}$ and $\mathfrak{I}\{\cdot\}$, respectively, and expectation by $\mathbb{E}[\cdot]$. A circularly symmetric complex Gaussian random variable is denoted by $X \sim \mathcal{CN}(0,\sigma^2)$, a real Gaussian by $X \sim \mathcal{N}(0,\sigma^2)$, and a uniform distribution over $[a,b]$ by $X \sim \mathcal{U}[a,b]$. Finally, $\mathrm{diag}(\boldsymbol{x})$ denotes a diagonal matrix with diagonal entries  $\boldsymbol{x}$. 

\section{System model} \label{sec:systemmodel}
We consider a single-cell MU-MIMO downlink communication system, as depicted in Fig.~\ref{fig:systemmodel}, where the BS contains an AAS with $M$ antenna-integrated radios and serves $K$ single-antenna UEs. The BBU is connected to an AAS through a limited-capacity fronthaul link, which is modeled as a finite-resolution quantizer. Each transmitted signal vector is computed as the product of a precoding matrix $\boldsymbol{P}$ and a data symbol vector $\boldsymbol{s}$. While the precoding matrix remains fixed during a transmission interval, the data symbols are changed at the symbol rate. The BBU encodes data and generates the precoding matrix based on CSI. Then, the precoding matrix and encoded data are transmitted to the AAS.  
As data symbols are bit sequences from a channel code, we can transmit them over the fronthaul without quantization errors, as they are already quantized. In contrast to earlier C-RAN studies, massive MIMO systems with large antenna arrays can make precoding signaling a dominant contributor to the fronthaul load rather than data. Hence, quantization-aware precoding design is a fundamental system component. We can then map the data symbols obtained from the BBU to modulation symbols at the AAS. However, the precoding matrix is quantized due to the constrained fronthaul. The quantized precoding matrix is then multiplied by the UEs' data symbols at the AAS, and the product is transmitted wirelessly.  We numerically observed that under the same quantization resolution, the quantized precoding signal method gives slightly lower MSE. However, the proposed separate signaling scheme requires significantly less fronthaul capacity. Therefore, in massive MIMO systems with large antenna arrays, the proposed method offers a better performance–fronthaul trade-off.

  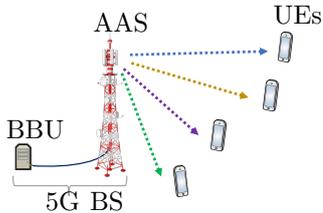
\begin{figure}[t!]
  \centering
   \begin{overpic}[scale=0.2,unit=1mm]{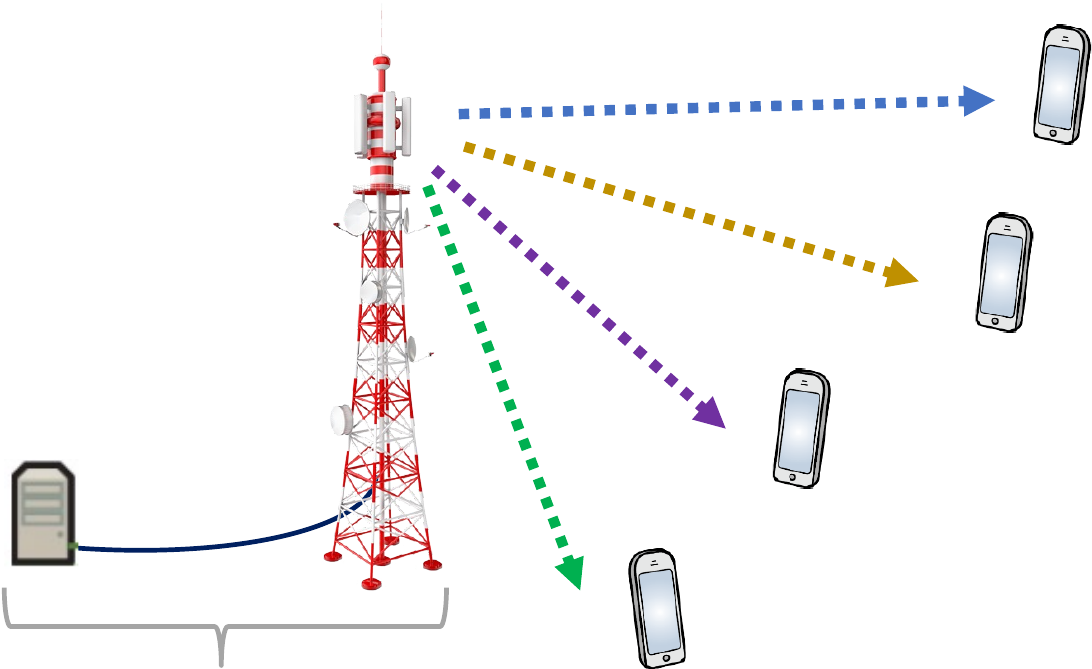}
  \put(-2,21){BBU}%
  \put(29,59){AAS}%
  \put(93,62){UEs}%
    \put(12,-5){5G BS}%
   \end{overpic}
    \vspace{2mm}
\caption{Downlink MU-MIMO system with constrained fronthaul capacity. 
}
\label{fig:systemmodel}
\vspace{-5mm}
\end{figure}

\begin{remark}
In fronthaul-constrained large mMIMO systems, the fronthaul load can be substantially reduced by the proposed separate precoding matrix and data symbols transmission. Let $N_{\mathrm{precoding}}$ denote the number of quantization bits per real dimension used for the precoding coefficients, and  $\tau_{\rm s}$ the number of data symbol vectors per block. The required fronthaul capacity is then $C_{\mathrm{separate}} = 2 M K N_{\mathrm{precoding}} + \tau_{\rm s} K \eta \quad \text{bits per block}$, where $\eta$ denotes the spectral efficiency in bits per data symbol. If the BBU instead computes and forwards the $M$-dimensional precoded signals directly, the required capacity becomes $C_{\mathrm{joint}} = M \tau_{\rm s} \eta N_{\mathrm{overSamp}}  \quad \text{bits per block}$, where $N_{\mathrm{overSamp}}$ represents the oversampling when transmitting I/Q samples. For example, consider $16$-QAM ($\eta = 4$~bit/symbol), $M=16$ antennas, $K=4$ users, $\tau_{\rm s}=100$ symbols per block, and $N_{\mathrm{precoding}} = N_{\mathrm{overSamp}} = 3$. This yields $C_{\mathrm{separate}} = 1984$ bits and $C_{\mathrm{joint}} = 19200$ bits per block, demonstrating that separate signaling of the precoding matrix and data symbols reduces the fronthaul load by approximately one order of magnitude. Moreover, in very large massive MIMO systems, the fronthaul cost associated with precoding transmission may even exceed that of data transmission, underscoring the importance of efficient precoding quantization.
\end{remark}

Before presenting the proposed scheme in Section~\ref{sec3}, we introduce the model for limited-capacity fronthaul, present the downlink communication, describe the imperfect CSI model, and formulate the sum-rate maximization problem.
\vspace{-2mm}
\subsection{Limited-Capacity Fronthaul Model}
We model the constrained fronthaul capacity through quantization. In practice, uniform quantization is often used, so we assume that our quantizer function $\mathcal{Q}(\cdot) : \mathbb{C} \to \mathcal{P}$ is a symmetric uniform quantization with step size $\Delta$. The quantizer function $\mathcal{Q}(\cdot)$  is uniquely characterized by the set of real-valued quantization labels $\mathcal{L} = \{ l_z : z=1, \ldots, L \}$, where  $L=|\mathcal{L}|$ is the number of quantization levels  $L=|\mathcal{L}|$,  $\mathbb{L} =\log_2 (L)$ denotes the number of quantization bits per real dimension, and for index $z$ we have $\mathcal{L}(z)=l_z$. Each entry of the quantization labels $\mathcal{L}$ is defined as 
\begin{equation}
    l_z  = \Delta  \left( z-1 -  \frac{L-1}{2}  \right), \quad z=1, \ldots, L.
\end{equation}
The quantizer step size $\Delta$  should be selected to minimize the distortion between the input and its quantized output. Since the optimal $\Delta$ depends on the statistical distribution of the input, which in this case is influenced by the precoding scheme and the channel model, we adopt a maximum-entropy approach. Specifically, we assume the per-antenna input to the quantizer follows a complex Gaussian distribution $\mathcal{CN}(0, \frac{q}{KM})$, where the variance is chosen to satisfy the power constraint in \eqref{eq:power constrain}. We use the corresponding optimal step size for a normal distribution, which was numerically derived in \cite{hui2001unifquantized}.

Furthermore, we let the set of quantization thresholds $\mathcal{T} = \{ \tau_0, \ldots, \tau_L \}$, where $-\infty = \tau_0 < \tau_1 < \ldots < \tau_{(L-1)} < \tau_{L} =\infty$, specify the set of the $L + 1$ quantization thresholds. For uniform quantizers, the thresholds are 
\begin{equation}
    \tau_z =  \Delta  \left( z- \frac{L}{2}  \right), \quad z=1, \ldots, L-1.
\end{equation}
The quantizer maps an input $\chi \in \mathbb{C} $ to the quantized output $\mathcal{Q}(\chi) = l_{R} + jl_{I} \in \mathcal{P}$, where $\mathfrak{R}\{ \mathcal{Q}(\chi) \}=l_{R} \in  [\tau_o,\tau_{o+1})$ and $\mathfrak{I}\{ \mathcal{Q}(\chi)=l_{I} \} \in  [\tau_l,\tau_{l+1})$, and $\mathcal{P}$ is the set of fronthaul quantization alphabet and is defined as 
\begin{equation}
   \mathcal{P} = \{ l_{R} + jl_{I} : l_{R},l_{I} \in \mathcal{L}  \}.
   \label{eq:quantizationset}
\end{equation}
It is assumed that the same quantization alphabet is applied to both the real and imaginary components. Note that $\mathcal{P}$ becomes the complex-number set $\mathbb{C}$ in the case of infinite resolution.

\subsection{Downlink Communications}
The BBU uses its available CSI to select the downlink precoding matrix $\boldsymbol{P}$. As this work mainly focuses on managing quantization errors in a precoding matrix that maximizes the sum rate, we first consider perfect CSI and disregard any potential transceiver hardware impairments. Subsequently, we illustrate how the proposed algorithms can also be utilized in scenarios with imperfect CSI. The transmitted data symbol to the UE $k$ is denoted by $s_k$, which has zero mean and normalized unit power. The corresponding channel vector  $\boldsymbol{h}^\mathrm{T}_k \in \mathbb{C}^{1 \times M}$ represents a narrowband channel 
and could be one subcarrier of a multi-carrier system, in which case the algorithms developed in this paper are applied independently on each subcarrier. The received signal at UE $k$ is
\begin{equation}
    y_k = \boldsymbol{h}^\mathrm{T}_k {\boldsymbol{p}}_k s_k + \sum_{i=1,i\ne k}^K \boldsymbol{h}^\mathrm{T}_k {\boldsymbol{p}}_i s_i + n_k,
\end{equation}
where ${\boldsymbol{p}}_k\in \mathcal{P}^{M}$ is the quantized linear precoding vector for UE $k$, $\mathcal{P}^M$ denotes the set of all length-$M$ vectors whose elements belong to the quantization alphabet $\mathcal{P}$, and $n_k \sim \mathcal{CN}(0,N_0)$ represents the independent additive complex Gaussian receiver noise with power $N_0$. The transmitted signal must adhere to the power constraint
\begin{equation}
\label{eq:sum_power_constraint}
   \mathbb{E}\left\{\|\vect{p}_k s_k\|_2^2\right\} = \sum_{k=1}^K \|\vect{p}_k\|_2^2 \leq q,
\end{equation}
where $q$ denotes the average maximum transmit power of the downlink signals. The quantized precoding matrix $\boldsymbol{P}$ and the data symbols vector $\boldsymbol{s}$ are sent separately over the fronthaul. 

For later use, we define the total received signal as $\boldsymbol{y} = [y_1,\ldots,y_K]^\mathrm{T}$, the data symbols vector as $\boldsymbol{s} = [s_1, \ldots, s_K]^\mathrm{T}\in \mathcal{O}^{K}$ ($\mathcal{O}$ is a finite set of constellation points such as a QAM alphabet), the channel matrix as  $\boldsymbol{H} = [\boldsymbol{h}_1,\ldots,\boldsymbol{h}_K]^\mathrm{T} \in \mathbb{C}^{K\times M}$, and the precoding matrix as $\boldsymbol{P} = [\boldsymbol{p}_1,\ldots,\boldsymbol{p}_K]$. It is worth mentioning that the precoded signal vector $\boldsymbol{x}$ is calculated at the AAS as $ \boldsymbol{x} = \alpha \boldsymbol{P} \boldsymbol{s}$, where the scaling factor $\alpha = \sqrt{q/ \mathrm{tr}(\boldsymbol{P}\boldsymbol{P}^\mathrm{H})}$ is computed at the AAS. This scaling factor ensures that the  $\mathbb{E}[\left \Vert \boldsymbol{x} \right \|_2^2] = q$ so that the maximum power is always utilized, despite the finite resolution.

\vspace{-3mm}

\subsection{Channel Estimation and Uplink Quantization} \label{sec:estimation}

The BBU requires CSI to determine the downlink precoding matrix. The CSI is represented by the estimated channel matrix $ \hat{\boldsymbol{H}} \in \mathbb{C}^{K \times M}$, where each element is denoted as $ \hat{h}_{k,m} $. Since our main objective is to analyze the impact of fronthaul quantization, we mainly assume perfect CSI, meaning $ \hat{\boldsymbol{H}} = \boldsymbol{H} $. However, the proposed algorithms can also operate when the BBU has imperfect CSI but evaluate the objective function under the assumption of perfect CSI. This approach is commonly adopted in precoding design with imperfect CSI \cite{jindal2006mimo}.

In this section, we describe the model of CSI imperfections used in our numerical results. We consider a time-division duplexing (TDD) system, where the channel is first estimated during the uplink phase and then applied in the downlink by utilizing channel reciprocity. To facilitate the BBU with CSI, $\tau_\mathrm{p}$  samples are allocated for uplink pilot-based channel estimation within each coherence block, allowing for $\tau_\mathrm{p}$ mutually orthogonal pilot sequences.  These pilot sequences are assigned to different UEs. The deterministic pilot sequence corresponding to UE $k$ is represented as $ \boldsymbol{\phi}_{k} \in \mathbb{C}^{\tau_p} $, satisfying the condition $ \|\boldsymbol{\phi}_{k} \|_2^2= \tau_p $.  The received pilot signal vector $\boldsymbol{Y}_k^{\text{Pilot}} \in \mathbb{C}^{M \times \tau_p }$ from the $k$-th UE at the AAS is given by
\begin{equation}
\boldsymbol{Y}^{\text{Pilot}}_k = \sqrt{q^\mathrm{U}_k} \boldsymbol{h}_k\boldsymbol{\phi}_{k}^\mathrm{T}  + {\boldsymbol{N}}^{\text{Pilot,U}}_k,
\end{equation}
where $q^\mathrm{U}_k$ is the transmit power used by UE $k$ and ${\boldsymbol{N}}^{\text{Pilot,U}}_k \in \mathbb{C}^{M \times \tau_p}$ is the noise matrix at the BS with i.i.d. $\mathcal{CN}(0,\sigma^2_{\text{U}})$-entries. To estimate the channel $\boldsymbol{h}_k$, the BBU multiplies $\boldsymbol{Y}^{\text{Pilot}}_k$ with the conjugate of the UE's pilot sequence ${{\boldsymbol{\phi}}}_k^*$ to obtain $\boldsymbol{y}^{\text{Pilot}}_{k} = \boldsymbol{Y}^{\text{Pilot}}_k \boldsymbol{\phi}^*_k$. The LS estimate 
$\hat{\boldsymbol{h}}_k$ is derived by minimizing  $\| \boldsymbol{y}^{\text{Pilot}}_{k} - \sqrt{q^\mathrm{U}_k} \tau_p\hat{\boldsymbol{h}}_k \|_2^2 $ \cite{Kay1993a}, resulting in
\begin{equation}  \label{eq:channelMLuplink} \hat{\boldsymbol{h}}_k = \frac{1}{\sqrt{q^\mathrm{U}_k} \tau_p}\boldsymbol{y}^{\text{Pilot}}_{k}.
\end{equation}

This model considers only the estimation errors caused by the limited uplink SNR while neglecting potential errors from uplink fronthaul quantization.  Now, given the possibility of channel estimation at AAS or BBU, there are two scenarios. In \cite{masoumi2019performance}, it is shown that if an AAS estimates the channel and then sends quantized CSI over a limited fronthaul link, the uplink sum rate is higher. Motivated by this, we consider the scenario where the AAS forwards the estimated uplink channel $\hat{\boldsymbol{H}}^\mathrm{T} = [\hat{\boldsymbol{h}}_1,\ldots,\hat{\boldsymbol{h}}_K] \in \mathbb{C}^{M\times K}$ to the BBU that exploits channel reciprocity under TDD operation to perform downlink precoding. The fronthaul transmission of the channel estimates introduces quantization errors. To model this effect, we adopt the additive quantization noise model (AQNM), where the quantization distortion is represented as additive white noise applied independently to each channel entry. Accordingly, the quantized CSI available at the BBU is modeled as
\begin{equation}
  \hat{\mathbf{H}}_{\rm Q}^\mathrm{T} = (1 - \eta_{\rm H})\hat{\mathbf{H}}^\mathrm{T} + \mathbf{N}_{\mathrm{Q},U},
    \label{eq:CQCSI}
\end{equation}
where $\eta_{\rm H}$ is the AQNM distortion factor, which is a monotonically decreasing function of $B_{\rm H}$ (bits per complex entry). For the Gaussian input distribution that we get under Rayleigh fading, the value of $\eta_{\rm H}$ for $B_{\rm H} \le 5$ is listed in Table~\ref{tab:eta_h} and for  $B_{\rm H} > 5$ it can be tightly approximated by $\eta_{\rm H} = \frac{\pi\sqrt{3}}{2} \cdot 2^{-2B_{\rm H}}$ \cite{fan2015uplink}. The quantization noise  $\mathbf{N}_{\mathrm{Q},U}$ in \eqref{eq:CQCSI} is uncorrelated, and has zero mean, and covariance $ \eta_{\rm H}(1-\eta_{\rm H})\mathbb{E}[|\hat{h}^\mathrm{T}_{m,k}|^2])$ per entry, for $m=1,\ldots,M$ and $k=1,\ldots,K$. While the quantization noise is, in practice, slightly correlated across antennas and UEs, prior results show that if the number of quantization bits is practically large, modeling it as independent and identically distributed (i.i.d.) yields accurate performance predictions.

\begin{table}[t!]
\centering
\caption{AQNM distortion $\eta_{\rm H}$  vs. $B_{\rm H}$.}
\label{tab:eta_h}
\begin{tabular}{|c|ccccc|}
\hline
$B_{\rm H}$      & $1$      & $2$       & $3$       & $4$        & $5$      \\
\hline
$\eta_{\rm H}$  & $0.3634$ & $0.1175$  & $0.03454 $& $0.009497$ & 0.002499  \\
\hline
\end{tabular}
\end{table}

 In Section \ref{sec:numericalimperfect}, we evaluate the effect of uplink quantization on the system's sum rate.

The impact of these quantization errors on precoding at high SNR depends on the resolution of pilot signal quantization. However, since uplink data reception generally requires higher quantization resolution than downlink transmission, it is reasonable to assume that pilot signals are also received with relatively high resolution. In the numerical results, we evaluate how uplink quantization affects the system's sum rate performance.

\subsection{Problem Formulation}
The achievable rate is $\log_2 (1 + \mathrm{SINR}_k (\boldsymbol{P}))$, where
the signal-to-interference-plus-noise-ratio (SINR) depends on the precoding matrix $\boldsymbol{P}$ as   
   \begin{equation}   
   \mathrm{SINR}_k(\boldsymbol{P}) = \frac{\big| \boldsymbol{h}^\mathrm{T}_k {\boldsymbol{p}}_k \big|^2}{ \sum_{i=1,i\ne k}^K \big| \boldsymbol{h}^\mathrm{T}_k {\boldsymbol{p}}_i \big|^2 +\bar{N_0}},
\end{equation}
where $\bar{N_0}=N_0/\alpha$.
Our objective is to maximize the sum rate of the proposed downlink communication system under a maximum transmit power constraint at the BS, defined in \eqref{eq:sum_power_constraint}. We define this problem as
\begin{maxi!}[2]
	  {\boldsymbol{P}\in \mathcal{P}^{M \times K} }{\sum_{k=1}^K  \log_2 \Big(1 + \mathrm{SINR}_k (\boldsymbol{P})\Big)}{}{\mathbb{P}_{1}: \ \ } \label{eq:firstproblem} \addConstraint{\mathrm{tr}(\boldsymbol{P}\boldsymbol{P}^\mathrm{H}) \le q,} \label{eq:power constrain}
\end{maxi!}
where the optimization variable is
 $\boldsymbol{P}$ with elements  $p_{m,k} \in \mathcal{P}$ for $k=1, \ldots ,K$ and $m=1, \ldots ,M$. Here, $\mathcal{P}^{M \times K}$ denotes the set of $M \times K$ matrices whose entries belong to $\mathcal{P}$, forming the quantization codebook.
Problem $\mathbb{P}_{1}$ is not convex, as the utility is non-concave and the search space is discrete, so it is hard to find the globally optimal solution \cite{christensen2008weighted}, as will be discussed in Section~\ref{sec3}. Therefore, in this paper, we aim to achieve a locally optimal solution by reformulating problem $\mathbb{P}_{1}$ into a more tractable form.

\section{Proposed WMMSE Algorithm and Optimized Quantized Precoding}\label{sec3}

Inspired by the classical WMMSE approach \cite{shi2011iteratively}, we will rewrite $\mathbb{P}_1$ as an equivalent iterative WMMSE problem for which a local optimum can be found through alternating optimization. This equivalence allows us to decompose the primary problem into a sequence of subproblems, which can be solved alternately through an iterative process.

Let $\hat{s}_k = \beta_{k}{y}_k$ denote the estimate at UE $k$ of the transmitted data symbol $s_k$. It is obtained from the received signal $y_k$ using the receiver gain ${\beta}_k \in \mathbb{C}$ (also known as the precoding factor \cite{jacobsson2017quantized}). For a given receiver gain, the MSE in the data detection at UE $k$ becomes
\begin{align}
    e_k(\boldsymbol{P}, {\beta}_k) &=  \mathbb{E} \left[   | {s}_k -  \hat{{s}}_k  |^2 \right]    \nonumber  
    \\ & =
    \left|\beta_{k}\right|^2\left(\left|\boldsymbol{h}^\mathrm{T}_k {\boldsymbol{p}}_k\right|^2+\sum\limits_{i=1,i\ne k}^K\left|\boldsymbol{h}^\mathrm{T}_k {\boldsymbol{p}}_i\right|^2+\bar{N_0}\right) \nonumber \\& \quad -2\Re\left(\beta_{k}\boldsymbol{h}^\mathrm{T}_k {\boldsymbol{p}}_k \right)+1 \label{eq:mseUE}.
\end{align}
The MSE expression in \eqref{eq:mseUE} is a convex function of $\beta_k$. This convexity enables us to simplify problem  $\mathbb{P}_1$ by reformulating it into a more tractable form.

The following proposition provides a reformulation of the sum rate maximization problem $\mathbb{P}_1$ as a weighted sum MSE minimization problem, adapted from \cite{Shi2011}.

\begin{proposition}
By defining the auxiliary weight $d_k \geq 0$, 
the sum rate maximization problem $\mathbb{P}_1$ is equivalent to the weighted sum MMSE problem 
\begin{mini!}[2]
  {\boldsymbol{P}\in \mathcal{P}^{M \times K} , \boldsymbol{\beta},\boldsymbol{d} }{\sum_{k=1}^K  \Big( d_k e_k(\boldsymbol{P}, {\beta}_k) - \log_2 (d_k)\Big) \label{eq:wmmse}}{\label{probelm2}}{\mathbb{P}_{2}: \ } \addConstraint{\mathrm{tr}(\boldsymbol{P}\boldsymbol{P}^\mathrm{H}) \le q,} \nonumber
\end{mini!}
where $\boldsymbol{\beta} = [\beta_1,\ldots,\beta_K]^\mathrm{T}$ is a vector containing all receiver gains and $\boldsymbol{d} = [d_1,\ldots,d_K]^\mathrm{T}$ is a vector containing all the UE weights in the weighted MSE. Problem $\mathbb{P}_2$ is equivalent to $\mathbb{P}_1$ in the sense that the optimal $\boldsymbol{P}$ is the same for both problems. Note that the cost function in \eqref{eq:wmmse} is convex in each individual optimization variable, which is the key reason for considering this equivalent problem formulation.
\end{proposition}

\begin{IEEEproof}
Given the precoding vectors ${\boldsymbol{p}}_k $, we  can select the value of $\beta_k$ that minimizes the MSE for given $\boldsymbol{P}$ as
\begin{equation}
    \bar{\beta}_k (\boldsymbol{P})= \frac{ (\boldsymbol{h}^\mathrm{T}_k {\boldsymbol{p}}_k )^{*}}{\left|\boldsymbol{h}^\mathrm{T}_k {\boldsymbol{p}}_k\right|^2+\sum_{i=1,i\ne k}^K\left|\boldsymbol{h}^\mathrm{T}_k {\boldsymbol{p}}_i\right|^2 + \bar{N_0}}.
    \label{eq:beta}
\end{equation}

Moreover, the objective function  \eqref{eq:wmmse} is convex with respect to $d_k$, and its optimal value can be determined by computing its derivative as follows:
 \begin{align}
   \bar{d}_k &= \frac{1}{\ln (2) e_k(\boldsymbol{P}, \bar{\beta}_k (\boldsymbol{P}))} \nonumber \\ & = \frac{1}{\ln (2) }\left( 1 + \frac{\left|\boldsymbol{h}^\mathrm{T}_k {\boldsymbol{p}}_k\right|^2}{\sum_{i=1,i\neq k}^K \left|\boldsymbol{h}^\mathrm{T}_k {\boldsymbol{p}}_i \right|^2 +\bar{N_0}} \right) . \label{eq:weightsmse}
\end{align}
After substituting the optimal weights  \eqref{eq:weightsmse} into the objective function \eqref{eq:wmmse},  the optimization problem for the precoding $\boldsymbol{P}$ can be expressed as 
    \begin{mini}|l|
	  {\boldsymbol{P}\in \mathcal{P}^{M \times K} }{- \sum_{k=1}^K  \log_2 \left(\frac{1}{e_k(\boldsymbol{P}, \bar{\beta}_k (\boldsymbol{P}))} \right) }{}{}\label{eq:invererror}
     \end{mini}
subject to \eqref{eq:power constrain}.  If we substitute the optimal receiver gain \eqref{eq:beta} into the MSE expression \eqref{eq:mseUE}, we will have 
\begin{align}
  & \frac{1}{e_k(\boldsymbol{P}, \bar{\beta}_k (\boldsymbol{P}))} = \frac{\sum_{i=1}^K \left|\boldsymbol{h}^\mathrm{T}_k {\boldsymbol{p}}_i \right|^2 +\bar{N_0}}{\sum_{i=1,i\neq k}^K \left|\boldsymbol{h}^\mathrm{T}_k {\boldsymbol{p}}_i \right|^2 +\bar{N_0}} \nonumber \\
  & = 1 + \frac{\left|\boldsymbol{h}^\mathrm{T}_k {\boldsymbol{p}}_k\right|^2}{\sum_{i=1,i\neq k}^K \left|\boldsymbol{h}^\mathrm{T}_k {\boldsymbol{p}}_i \right|^2 +\bar{N_0}} = 1+ \mathrm{SINR}_k(\boldsymbol{P}).
\end{align} 
So \eqref{eq:wmmse}  becomes $K-\sum_{k=1}^{K} \log_2(1+\mathrm{SINR}_k(\boldsymbol{P}))$. Therefore, problem $\mathbb{P}_{2}$ is equivalent to problem $\mathbb{P}_{1}$.
\end{IEEEproof}

For fixed ${\beta}_{k}$  and ${d}_k$ (e.g., calculated as in \eqref{eq:beta} and  \eqref{eq:weightsmse}), we have the WMMSE problem
\begin{mini!}[2]
	  {\boldsymbol{P}\in \mathcal{P}^{M \times K}}{\sum_{k=1}^K   {d}_k e_k(\boldsymbol{P}, {\beta}_k) \label{eq:sum-SE-maximization-weighted-MMSE-subproblem}}{\label{probelm3}}{\mathbb{P}_{3}: \ } \addConstraint{\mathrm{tr}(\boldsymbol{P}\boldsymbol{P}^\mathrm{H}) \le q,} \nonumber
\end{mini!}
which is an integer convex problem. 
By iterating between updating 
$\beta_k$ using \eqref{eq:beta}, $d_k$ using \eqref{eq:weightsmse}, and $\boldsymbol{P}$ by solving $\mathbb{P}_3$, we obtain a block coordinate descent algorithm that will converge to a stationary point (for the same reasons as in \cite{shi2011iteratively}). Algorithm~\ref{Alg:WMMSE} summarizes the presented summary of the proposed iterative weighted sum MSE minimization approach.

We can initialize the algorithm using any precoding matrix, including those obtained using classical infinite-resolution precoding schemes.
The Wiener filtering (WF) precoding scheme (also known as regularized zero-forcing) is the most desirable one in this context since it can be derived by minimizing the sum MSE \cite{joham2005linear}. Hence, we suggest setting the initial  precoding matrix as $\overline{\vect{P}}^{(0)} = \boldsymbol{H}^\mathrm{H} (\boldsymbol{H}\boldsymbol{H}^\mathrm{H}+\frac{KN_0}{q}\boldsymbol{I}_K )^{-1}$.

\begin{algorithm}[t!]
\caption{Iterative WMMSE algorithm for solving problem $\mathbb{P}_2$.}
\label{Alg:WMMSE}
\begin{algorithmic}[1]
\STATEx {\textbf{Inputs}: $\vect{H}$, $K$, $q$, $N_0$,  Maximum number of iterations $N_{\text{max}}$, convergence threshold $\epsilon>0$ }
\STATE{Initialize the precoding matrix $\overline{\vect{P}}^{(0)} = [\overline{\vect{p}}_1^{(0)},\ldots,\overline{\vect{p}}_K^{(0)}] = \boldsymbol{H}^\mathrm{H} (\boldsymbol{H}\boldsymbol{H}^\mathrm{H}+\frac{KN_0}{q}\boldsymbol{I}_K )^{-1}$}
\STATE{Compute $\overline{\beta}_k^{(0)} = \frac{\left( \vect{h}_k^{\mathrm{T}} \overline{\vect{p}}^{(0)}_k\right)^*}{|\vect{h}_k^{\mathrm{T}} \overline{\vect{p}}^{(0)}_k|^2 + \sum_{i=1,i\neq k}^K |\vect{h}_k^{\mathrm{T}} \overline{\vect{p}}^{(0)}_i|^2 + \bar{N_0}},$ for $k = 1,\ldots,K$}
\STATE{Compute  $\overline{\vect{d}}^{(0)} = [\overline{d}_1^{(0)},\ldots,\overline{d}_K^{(0)}]^{\mathrm{T}}$, } where $\overline{d}_k^{(0)} = \frac{1}{\ln (2) }\Big( 1 + \frac{\left|\boldsymbol{h}^\mathrm{T}_k \overline{\vect{p}}^{(0)}_k\right|^2}{\sum_{i=1,i\neq k}^K \left|\boldsymbol{h}^\mathrm{T}_k \overline{\vect{p}}^{(0)}_i \right|^2 +\bar{N_0}} \Big)$ 
\STATE{Set $\varepsilon \gets \epsilon + 1$, $n \gets 0$}
\STATE{Find $\vect{e}^{(0)}  = [e_1^{(0)},\ldots,e_K^{(0)}]^{\mathrm{T}}$, where $e_k^{(0)} =   \left|\overline{\beta}_k^{(0)}\right|^2\Big(\left|\boldsymbol{h}^\mathrm{T}_k \overline{\vect{p}}^{(0)}_k\right|^2+\sum\limits_{i=1,i\ne k}^K\left|\boldsymbol{h}^\mathrm{T}_k \overline{\vect{p}}^{(0)}_i\right|^2+\bar{N_0}\Big)  -2\Re\left(\overline{\beta}_k^{(0)}\boldsymbol{h}^\mathrm{T}_k \overline{\vect{p}}^{(0)}_k \right)+1 $}

\STATE{
Compute $f^{(0)} = \sum_{k=1}^K \overline{d}_k^{(0)} e_k^{(0)} - \log_2(\overline{d}_k^{(0)})$}
\WHILE{$\varepsilon > \epsilon$ and $n < N_{\mathrm{max}}$}
\STATE{$n \gets n + 1$}
\LongState{Solve problem $\mathbb{P}_3$}
\FOR{$k = 1$ to $K$}
\STATE{Obtain  $\{\overline{\beta}_k^{(n)}\}$ from \eqref{eq:beta}}
\STATE{Obtain $\{\overline{d}_k^{(n)}\}$ from \eqref{eq:weightsmse}}
\LongState{Compute $\{e_k^{(n)}\}$ using \eqref{eq:mseUE}}
\ENDFOR
\STATE{Compute $f^{(n)} = \sum_{k=1}^K \overline{d}_k^{(n)} e_k^{(n)} - \log_2(\overline{d}_k^{(n)})$}
\STATE{Compute convergence gap $\varepsilon = |f^{(n)} - f^{(n-1)}|$}
\ENDWHILE
\STATEx{ \textbf{Outputs:} $\overline{\vect{p}}_k^{(n)}~k = 1,\ldots, K,$} 
 \end{algorithmic}
\end{algorithm}

\begin{remark} \label{rem:weighted} By utilizing Algorithm~\ref{Alg:WMMSE}, we can efficiently solve the sum rate maximization problem $\mathbb{P}_1$ locally. By following the same iterative steps, we can aalso xtend this approach to solve the weighted sum rate maximization problem:
\begin{maxi!}[2]
	  {\boldsymbol{P}\in \mathcal{P}^{M \times K} }{\sum_{k=1}^K  u_k \log_2 \Big(1 + \mathrm{SINR}_k (\boldsymbol{P})\Big)}{}{ \ \ } \label{eq:firstproblem2} \addConstraint{\mathrm{tr}(\boldsymbol{P}\boldsymbol{P}^\mathrm{H}) \le q,} \nonumber \addConstraint{u_k \ge 0, \quad k=1,\ldots,K,}
\end{maxi!}
where $u_k$ is the weight of UE~$k$.
This formulation is particularly relevant in systems where certain UEs have higher priority, such as QoS-constrained networks, fairness-aware scheduling, and differentiated service applications. The primary modification in this case lies in  \eqref{eq:weightsmse}, where the parameter $\bar{d}_k$ is updated to incorporate the UE-specific weights as $\bar{d}_k^{\text{new}} = u_k\bar{d}_k$, ensuring that UEs with higher weights receive greater emphasis in the optimization process.
\end{remark}

The main complexity in the proposed algorithm originates from solving $\mathbb{P}_3$. 
We rewrite the objective function \eqref{eq:sum-SE-maximization-weighted-MMSE-subproblem} as 
\begin{align}
 & \sum_{k=1}^K {d}_{k}e_{k}\left(\boldsymbol{P},{\beta}_k\right) = \mathbb{E} \left[  \left \Vert \sqrt{\mathrm{diag}(
 {\boldsymbol{d}}
 )} \Big (    \boldsymbol{s} -   \mathrm{diag}(
 {\boldsymbol{\beta}}) \boldsymbol{y} \Big) \right \|_2^2 \right],
 \label{eq:mmseform}
\end{align} 
and expand the expression in \eqref{eq:mmseform} as 
\begin{align}
&\mathbb{E} \left[  \left \Vert \sqrt{\mathrm{diag}(\boldsymbol{d})}    \boldsymbol{s} -  \sqrt{\mathrm{diag}(\boldsymbol{d})} \mathrm{diag}(\boldsymbol{\beta}) \boldsymbol{y} \right \|_2^2 \right] \nonumber  
\\ &= \mathrm{tr} \Big(\mathrm{diag}(\boldsymbol{d})  - \sqrt{\mathrm{diag}(\boldsymbol{d})}\boldsymbol{D} \boldsymbol{H}\boldsymbol{P}-  \sqrt{\mathrm{diag}(\boldsymbol{d})} \boldsymbol{P}^\mathrm{H}
 \boldsymbol{H}^\mathrm{H}\boldsymbol{D}^\mathrm{H} \nonumber \\& \quad + \boldsymbol{D}\boldsymbol{H}\boldsymbol{P}\boldsymbol{P}^\mathrm{H}\boldsymbol{H}^\mathrm{H}\boldsymbol{D}^\mathrm{H}   +  \bar{N_0} \boldsymbol{D} \boldsymbol{D}^\mathrm{H}  \Big), \label{eq:simplified} 
\end{align} 
where we introduce the notation $\boldsymbol{D}= \sqrt{\mathrm{diag}(\boldsymbol{d})} \mathrm{diag}(\boldsymbol{\beta})$.

We first notice that $\mathbb{P}_3$ is a so-called integer least-squares problem due to the \emph{uniform quantizer} that we chose. It can be solved using general-purpose methods, such as CVX \cite{grant2014cvx}, but instead of using a general-purpose solver that will have high complexity, we propose an efficient dedicated algorithm. First, we rewrite the objective
function \eqref{eq:simplified} using a Lagrange multiplier $\omega$ as 
\begin{align}
    \mathfrak{L}(\boldsymbol{P}, \boldsymbol \beta, \omega) &= \mathrm{tr} \Big(\mathrm{diag}(\boldsymbol{d})  - \sqrt{\mathrm{diag}(\boldsymbol{d})}\boldsymbol{D} \boldsymbol{H}\boldsymbol{P}  \nonumber \\ & -  \sqrt{\mathrm{diag}(\boldsymbol{d})} \boldsymbol{P}^\mathrm{H}
 \boldsymbol{H}^\mathrm{H}\boldsymbol{D}^\mathrm{H}   + \boldsymbol{D}\boldsymbol{H}\boldsymbol{P}\boldsymbol{P}^\mathrm{H}\boldsymbol{H}^\mathrm{H}\boldsymbol{D}^\mathrm{H}  \nonumber \\& +  \bar{N_0} \boldsymbol{D} \boldsymbol{D}^\mathrm{H}  \Big) + \omega \big( \mathrm{tr}(\boldsymbol{P}\boldsymbol{P}^\mathrm{H})-q \big). \label{eq:lagrangem}
\end{align}
When minimizing \eqref{eq:lagrangem} with respect to $\boldsymbol{P}$, we can drop the constant term $\mathrm{tr} (\mathrm{diag}(\boldsymbol{d}) )$ and have problem $\mathbb{P}_4$, given in \eqref{eq:lastsd} at the top of the next page. Although strong duality does not hold for the integer least-squares problem $\mathbb{P}_3$, $\mathbb{P}_4$ provides a good approximation. 
\begin{figure*}[ht!]
\begin{equation}\label{eq:lastsd}
    \mathbb{P}_{4}:
\underset{\omega \ge 0}{\text{maximize}}
\;
\underset{\boldsymbol{P}\in \mathcal{P}^{M \times K}}{\text{min}}
\;
\mathrm{tr}\Big(
\boldsymbol{P}^\mathrm{H}
\big(\boldsymbol{H}^\mathrm{H}\boldsymbol{D}^\mathrm{H}\boldsymbol{D}\boldsymbol{H}
+ \omega \boldsymbol{I}_M \big)
\boldsymbol{P}
- \sqrt{\mathrm{diag}(\boldsymbol{d})}\boldsymbol{D}\boldsymbol{H}\boldsymbol{P}
- \big(\sqrt{\mathrm{diag}(\boldsymbol{d})}\boldsymbol{D}\boldsymbol{H}\boldsymbol{P}\big)^\mathrm{H}
\Big)
- \omega q .
\end{equation}

\begin{mini}|l|
	  {\boldsymbol{p}_i\in {\mathcal{P}}^{M }, i=1,\ldots,K }{\sum_{i=1}^K \Big( \boldsymbol{p}_i^\mathrm{H} \left ( \boldsymbol{H}^\mathrm{H}\boldsymbol{D}^\mathrm{H} \boldsymbol{D}\boldsymbol{H}+\omega \boldsymbol{I}_M \right )\boldsymbol{p}_i -\boldsymbol{f}_i^\mathrm{T}\boldsymbol{p}_i - \left ( \boldsymbol{f}_i^\mathrm{T}\boldsymbol{p}_i\right )^\mathrm{H}\Big),}{ \label{eq:modifysd}}{\mathbb{P}_{5}: \ }
\end{mini}
\hrule
\end{figure*}
For a fixed value of $\omega$, and by vectorizing  $\mathbb{P}_4$ and using  $\boldsymbol{f} = \mathrm{vec}( (\sqrt{\mathrm{diag}(\boldsymbol{d})}\boldsymbol{DH})^\mathrm{T})$, we  obtain $\mathbb{P}_5$ in \eqref{eq:modifysd},
which finds a suboptimal precoding matrix and has $K$ separable objective functions that each only depend on one of the optimization variables.
This feature enables \emph{parallel} optimization of $\boldsymbol{p}_i$ for $i=1,\ldots,K$. Thus, in addition to the more efficient search strategy, the reformulation of problem $\mathbb{P}_5$ also significantly reduces the dimension of each subproblem \cite{khorsandmanesh2023optimized}. By defining $\hat{\boldsymbol{V}} = \boldsymbol{H}^\mathrm{H}\boldsymbol{D}^\mathrm{H} \boldsymbol{D}\boldsymbol{H}+ \omega \boldsymbol{I}_M$, we can obtain the equivalent formulation of each term of the objective function $\mathbb{P}_5$ as 
\begin{align}   \boldsymbol{p}_i^\mathrm{H} \hat{\boldsymbol{V}} \boldsymbol{p}_i -\boldsymbol{f}_i^\mathrm{T}\boldsymbol{p}_i - \left ( \boldsymbol{f}_i^\mathrm{T}\boldsymbol{p}_i\right )^\mathrm{H}  = \lVert \boldsymbol{c}_i - \boldsymbol{G}\boldsymbol{p}_i \rVert_2^2 - \boldsymbol{c}_i^\mathrm{H}\boldsymbol{c}_i, \label{eq:sdLS}
\end{align}
where $\boldsymbol{G} \in \mathbb{C}^{M \times M }$ is an upper triangular matrix that is obtained from the Cholesky decomposition $\hat{\boldsymbol{V}} = \boldsymbol{G}^\mathrm{H} \boldsymbol{G}$ and $\boldsymbol{c}_i = (\boldsymbol{f}_i^\mathrm{T}\boldsymbol{G}^{-1})^\mathrm{H}$. Using this notation, the sub-problem for optimizing $\mathbb{P}_5$ can be rewritten as
\begin{mini}|l|	  {\boldsymbol{p}_i\in {\mathcal{P}}^{M } }{\lVert \boldsymbol{c}_i - \boldsymbol{G}\boldsymbol{p}_i \rVert_2^2,}{ \label{eq:modifysd2}}{\mathbb{P}_{6}: \ }
\end{mini}
which is an integer least-squares problem due to the finite-resolution precoder. In problem $\mathbb{P}_{6}$, the precoding vector entries $\boldsymbol{p}_i$ are complex-valued, resulting in a search space that encompasses complex-valued alphabets. Noting the independence of the real and imaginary components of the set $\mathcal{P}$, problem $\mathbb{P}_{6}$ can be reformulated in an equivalent real-valued form. This transformation simplifies the computational process, as the two-dimensional search over complex precoding entries is reduced to a one-dimensional search. In other words, we utilize the
definitions
\begin{equation}
\begin{aligned} &\boldsymbol{c}_{i,\mathbb{R}} = \begin{bmatrix}
    \mathfrak{R}(\boldsymbol{c}_{i}) \\
    \mathfrak{I}(\boldsymbol{c}_{i})
    \end{bmatrix}, \quad
    \boldsymbol{p}_{i,\mathbb{R}} = \begin{bmatrix}
    \mathfrak{R}(\boldsymbol{p}_{i}) \\
    \mathfrak{I}(\boldsymbol{p}_{i})
    \end{bmatrix}, \\
    &\boldsymbol{G}_{\mathbb{R}} =
    \begin{bmatrix}
    \mathfrak{R}(\boldsymbol{G}) & -\mathfrak{I}(\boldsymbol{G})  \\
    \mathfrak{I}(\boldsymbol{G}) & \mathfrak{R}(\boldsymbol{G})
    \end{bmatrix}.
\end{aligned}
\end{equation}
to write a real-valued reformulation of problem $\mathbb{P}_{6}$. Since the $i$-th optimization in problem $\mathbb{P}_{6}$ can be performed independently of the other elements, we will omit the subscript $i$ for simplicity of notation. Moreover, for the rest of the paper, we utilize a real-domain formulation, and also omit the subscript $\mathbb{R}$. Then we want to solve the following rewritten format of the problem $\mathbb{P}_{6}$:
\begin{mini}|l|	  {\boldsymbol{p}\in {\mathcal{L}}^{2M } }{\lVert \boldsymbol{c} - \boldsymbol{G}\boldsymbol{p} \rVert_2^2,}{ \label{eq:sdreduced notation}}{\mathbb{P}_{7}: \ }
\end{mini}
and run it for each $i$ in the real domain. 
Under the adopted real-domain formulation, $\mathcal{L}^{2M}$ denotes the Cartesian product of the real-valued quantization alphabet $\mathcal{L} $ taken $2M$ times, since each complex precoding coefficient in $\mathcal{P}$ is formed from one real and one imaginary component drawn from $\mathcal{L}$.

The optimal value of $\omega$, which satisfies the power constraint \eqref{eq:power constrain} near equality, can be found via heuristic bisection search. Note that since this is a non-convex problem, we can not guarantee monotonicity or convergence to the global optimum. However, we added a checkpoint to heuristically guarantee convergence to the best point that is tested within the algorithm. The bisection algorithm will fix the updated $\omega$ only if the new precoding matrix has a higher sum rate. We presented this heuristic bisection search in Algorithm~\ref{alg:bisection}.

\begin{algorithm}[t!]
\caption{Heuristic Bisection Algorithm}
\label{alg:bisection}
\begin{algorithmic}[1]
\STATE Initialize $\omega_{\min}=1$, $\omega_{\max}=1$, $q$, and an initial precoding $\boldsymbol{P}$.
\STATE Set $f(\omega) \leftarrow \mathrm{tr}(\boldsymbol{P}(\omega)\boldsymbol{P}(\omega)^\mathrm{H}) - q$.
\REPEAT
\STATE Set $\omega \leftarrow (\omega_{\min}+\omega_{\max})/2$.
\STATE Calculate $f(\omega)$.
\STATE \textbf{if} $f(\omega_{\rm max})\cdot f(\omega) < 0$ \textbf{then}
\STATE  $\omega_{\rm min}\leftarrow \omega$
\STATE \textbf{if} $f(\omega_{\rm min})\cdot f(\omega) < 0$ \textbf{then}
\STATE  $\omega_{\rm max}\leftarrow \omega$
 \STATE \textbf{end if}
\STATE Solve the inner optimization problem $\mathbb{P}_{3}$ to obtain $\boldsymbol{P}(\omega)$.
    \STATE \textbf{if}  sum rate $\sum_{k=1}^K  \log_2 \Big(1 + \mathrm{SINR}_k (\boldsymbol{P}(\omega)\Big)$ is improved \textbf{then}
    \STATE \quad Update $\boldsymbol{P} \leftarrow \boldsymbol{P}(\omega)$.
   \STATE  \textbf{else} break
   \STATE \textbf{end if}
\STATE $\boldsymbol{P}(\omega)$ will only be updated if $\sum_{k=1}^K  \log_2 \Big(1 + \mathrm{SINR}_k (\boldsymbol{P}(\omega)\Big)$ is higher than previous round otherwise \textbf{break}
 \UNTIL{ $\left|\mathrm{tr}(\boldsymbol{P}\boldsymbol{P}^\mathrm{H}) - q  \right|$ is less than a threshold or the bisection interval becomes very small or the iterations becomes 100.}   

\end{algorithmic} 

\end{algorithm}

\subsection{Sphere Decoding-Based Precoding}

The search space of the problem $\mathbb{P}_{7}$ is a scaled finite subset of the infinite integer lattice. A technique that has previously been proposed as an efficient algorithm to solve closest lattice point problems in the Euclidean sense is called SD. We suggest using the Schnorr-Euchner SD method, a depth-first tree search algorithm, to solve problem $\mathbb{P}_{7}$ for a fixed value of $\omega$. We then make a bisection search over $\omega$ to find a value that gives a solution that satisfies the power constraint in \eqref{eq:power constrain} near equality.  We denote the optimal precoding solution by $\overline{\vect{p}}$. Note that parallel algorithms can be run for all $K$ precoding vectors. For the sake of completeness, we provide a brief overview of the proposed SD-based algorithm. More detailed explanations can be found in \cite{agrell2002closest}.

The tree search begins with an initial search radius set to infinity. Starting from the highest level $m = 2M$, the algorithm computes $  \left|{{c}}_{2M} - {{G}}_{2M,2M}{p}_{2M}\right|$ for all ${p}_{2M} \in \mathcal{L}$. The computed values are sorted, and the sorting indices are stored in $\vect{w}_{2M}$, such that $\vect{w}_{2M}[1]$ is the index $x$ corresponding to $l_{x} \in \mathcal{L}$, which is the closest to  $\check{{p}}_{2M}$, while $\vect{w}_{2M}[L]$ indicates the index of the farthest element. Then, the quantization label that minimizes this absolute value is selected as $\check{{p}}_{2M}= \mathcal{L}(\vect{w}_{2M}[1])$. The algorithm then proceeds to the next lower level and repeats the process, determining $\vect{w}_{2M-1}$ and $\check{{p}}_{2M-1}$ in the same way, given the value of $\check{{p}}_{2M}$. This process continues for all other levels, going from $2M$ to 1. At each level $m$, the selected value $\check{{p}}_{m}$ is given by 
\begin{mini}|l|	  {{p}_{m} \in \mathcal{L}}{\left| {\xi}_m -  {{G}}_{m,m}{p}_{m}\right|,}{  \label{eq:precoding_minimizer}}{\check{{p}}_{m}= \ }
\end{mini}
where
    ${\xi}_m = 
    {{c}}_m - {\vect{G}}_{m,m+1:2M}\check{\vect{p}}_{m+1:2M}$.
Once \eqref{eq:precoding_minimizer} is solved for all $2M$ levels, the resulting combination of precoding labels, $\check{\vect{p}}$, is stored as the temporary optimal solution ${\vect{p}}$. The radius given by
\begin{equation}
   {r}_1 =  \sum_{m=1}^{2M} \big| {{c}}_m - {\vect{G}}_{m,m:2M}\check{\vect{p}}_{m:2M}\big|^2,
\end{equation}
is chosen as the new search radius. Subsequently, the algorithm updates $\check{{p}}_2$ by selecting the next label,  $\mathcal{L}(\vect{w}_2[2])$, while keeping $\check{{p}}_3, \ldots, \check{{p}}_{2M}$ fixed. The updated value of $\check{{p}}_2$  is checked to ensure that the new radius
\begin{small}${r}_2 = \sum_{i=2}^{2M} \left| {{c}}_i - {\vect{G}}_{i,i:2M}{\vect{p}}_{i:2M}  \right|^2$\end{small} is within the updated search radius $r_{\mathrm{opt}}$.  $r_{\mathrm{opt}}$ denotes the current best squared distance metric and serves as the adaptive pruning radius, initialized as $\infty$. If ${r}_2 < r_{\mathrm{opt}}$, the algorithm proceeds to update $\check{{p}}_1$ by solving \eqref{eq:precoding_minimizer} again. Otherwise, the algorithm moves to an upper level and updates
$\check{{p}}_3$, fixing $\check{{p}}_4,\ldots,\check{{p}}_{2M}$. This procedure is repeated iteratively. After each successful update of $\check{{p}}_m$, that is when the update leads to ${r}_m < r_{\mathrm{opt}}$, the algorithm proceeds to the next lower level by computing $\check{{p}}_{m-1}$ using~\eqref{eq:precoding_minimizer}. If the update of $\check{{p}}_m$ is not successful, the index $m$ is increased by 1. The algorithm attempts to update $\check{{p}}_{m+1}$ by selecting an alternative label from $\mathcal{L}$, 
guided by the sorting indices at that level. This process continues as new combinations of entries are explored.
If a new combination $\check{{p}}_1,\ldots,\check{{p}}_{2M}$ is found that satisfies  ${r}_1 = \|{\vect{c}} - {\vect{G}} \check{\vect{p}}\|_2^2 < r_{\mathrm{opt}}$, the search radius is updated to $r_{\mathrm{opt}} = {r}_1$. 
At each level $m$, the next quantization label is selected using the sorting indices, ensuring efficient exploration of the solution space. If the updated radius at level $m$ exceeds $r_{\mathrm{opt}}$, the algorithm moves to the next level without further updates. This iterative process ensures that the optimal combination of precoding labels is progressively refined, leveraging the sorting indices to minimize computational complexity while adhering to the search radius constraint.
It is essential to emphasize the role of sorting indices in the update process of precoding entries. Specifically, assume that $\check{{p}}_m = \mathcal{L}(\vect{w}_m[i])$ at a given iteration. When the algorithm proceeds from level $m-1$ to level $m$, the precoding entry is updated to $\check{{p}}_m = \mathcal{L}(\vect{w}_m[i+1])$. If this new candidate fails to satisfy the optimal radius constraint, i.e., if ${r}_m > r_{\mathrm{opt}}$ for $\check{{p}}_m = \mathcal{L}(\vect{w}_m[i+1])$, then no further updates to $\check{{p}}_m$ are necessary for the current values of $\check{p}_{m+1}, \ldots, \check{p}_{2M}$. In other words, if ${r}_m > r_{\mathrm{opt}}$ for the $(i+1)$-th candidate, the same inequality will hold for all subsequent candidates $\check{p}_m = \mathcal{L}(\vect{w}_m[i^\prime])$ with $i^\prime > i+1$. The pseudo-code of this SD algorithm is provided in \cite[Algorithm~1]{khorsandmanesh2023optimized}.

\subsection{Heuristic Expectation Propagation-Based Precoding}

The proposed SD-based precoding approach is specifically designed for problem $\mathbb{P}_{7}$, yet its computational complexity increases significantly as the number of antennas $N$ and UEs $K$ grows. In this section, we discuss a lower complexity approach to solve this problem. Specifically, we leverage a message-passing-based algorithm called the expectation propagation (EP) algorithm, which can achieve near-maximum likelihood performance for signal detection problems with a complexity order similar to zero-forcing.  EP addresses the discrete constraints on the precoding vector $\vect{p}$ by relaxing its discrete prior distribution, induced by the constellation points, by a continuous Gaussian prior distribution. The likelihood function is also approximated by a Gaussian distribution, enabling tractable Gaussian posterior computations. This Gaussian posterior is then projected onto the discrete constellation space, and the Gaussian prior is refined based on this projection. Using the refined prior, the Gaussian posterior is recalculated in the next iteration. This process is repeated iteratively until convergence. It is important to note that some non-trivial modifications are required to apply the EP algorithm, used in MIMO symbol detection, to the formulation of the problem in \eqref{eq:sdreduced notation}. This is particularly due to the unavailability of the stochastic error variance/mismatch between the  variable $\vect{c}$ and $\vect{Gp}$.\footnote{The developed EP algorithm for solving the least-square problem $\lVert \boldsymbol{c} - \boldsymbol{G}\boldsymbol{p} \rVert_2^2$ requires knowledge of the statistical error variance, introduced by the additive white Gaussian noise at the receiver side.}

\subsubsection{Maximum A Posteriori Problem}

We rewrite the minimization problem in \eqref{eq:sdreduced notation} as a maximum a posteriori problem:
\begin{equation}\label{eq_MAP}
    \arg\max_{\vect{p} \in \mathcal{L}^{2M} }  p(\vect{p} | \vect{c}),
\end{equation}
where \begin{flalign} \label{eq_EP:Posterior_ori}
    & p(\vect{p}|\vect{c}) = \frac{p(\vect{c}|\vect{p}) }{ p(\vect{c})} \cdot p(\vect{p}) \propto \mathcal{N} \left( \vect{c}: \vect{G} \vect{p} , {\sigma}^2 \vect{I} \right)   {p(\vect{p})}, 
\end{flalign}
 $p(\vect{p}) $ is the a priori distribution of $\vect{p}$, and $p(\vect{c})$ is omitted as it does not depend on the distribution of $\vect{p}$.

Notice that to derive \eqref{eq_EP:Posterior_ori}, we impose that $\vect{\eta} = \vect{c} - \vect{G} \vect{p}$, where $\vect{\eta} \sim \mathcal{N}(\vect{0}, \sigma^2 \vect{I})$. 
Consequently, the likelihood distribution $p(\vect{c} | \vect{p})$ is Gaussian, given by $\mathcal{N}(\vect{c}; \vect{G} \vect{p}, \sigma^2 \vect{I})$. This assumption enables the utilization of low-complexity algorithms such as EP. In Fig.~\ref{fig:histogram}, we plot the histogram of $\vect{c}-\vect{G p }$ for the specified system configuration. The randomness in $\boldsymbol{c}-\boldsymbol{G}\boldsymbol{p} $ originates solely from the channel realization, since $\boldsymbol{p}$ is a deterministic function of the channel obtained via Lagrangian optimization, and the reported histogram is generated for a fixed Lagrange multiplier $\omega=1$.
Each data point in the histogram corresponds to a channel realization. The $\vect{p }$ is selected as columns of the initial  WF precoding matrix $\overline{\vect{P}}^{(0)} = \boldsymbol{H}^\mathrm{H} (\boldsymbol{H}\boldsymbol{H}^\mathrm{H}+\frac{KN_0}{q}\boldsymbol{I}_K )^{-1}$. We can see that $\vect{c}-\vect{G p }$ approximately follows the Gaussian distribution. The fitted normal probability density function (PDF) is plotted using the ksdensity function in MATLAB, which estimates the PDF of a random variable based on a finite set of samples using kernel smoothing techniques. From this observation, we conclude that applying the EP to the problem in \eqref{eq_MAP} is feasible if we can estimate the error variance $\sigma^2$.

\subsubsection{The EP Algorithm}
\begin{figure}[!t]
        \centering    \includegraphics[scale=0.25]{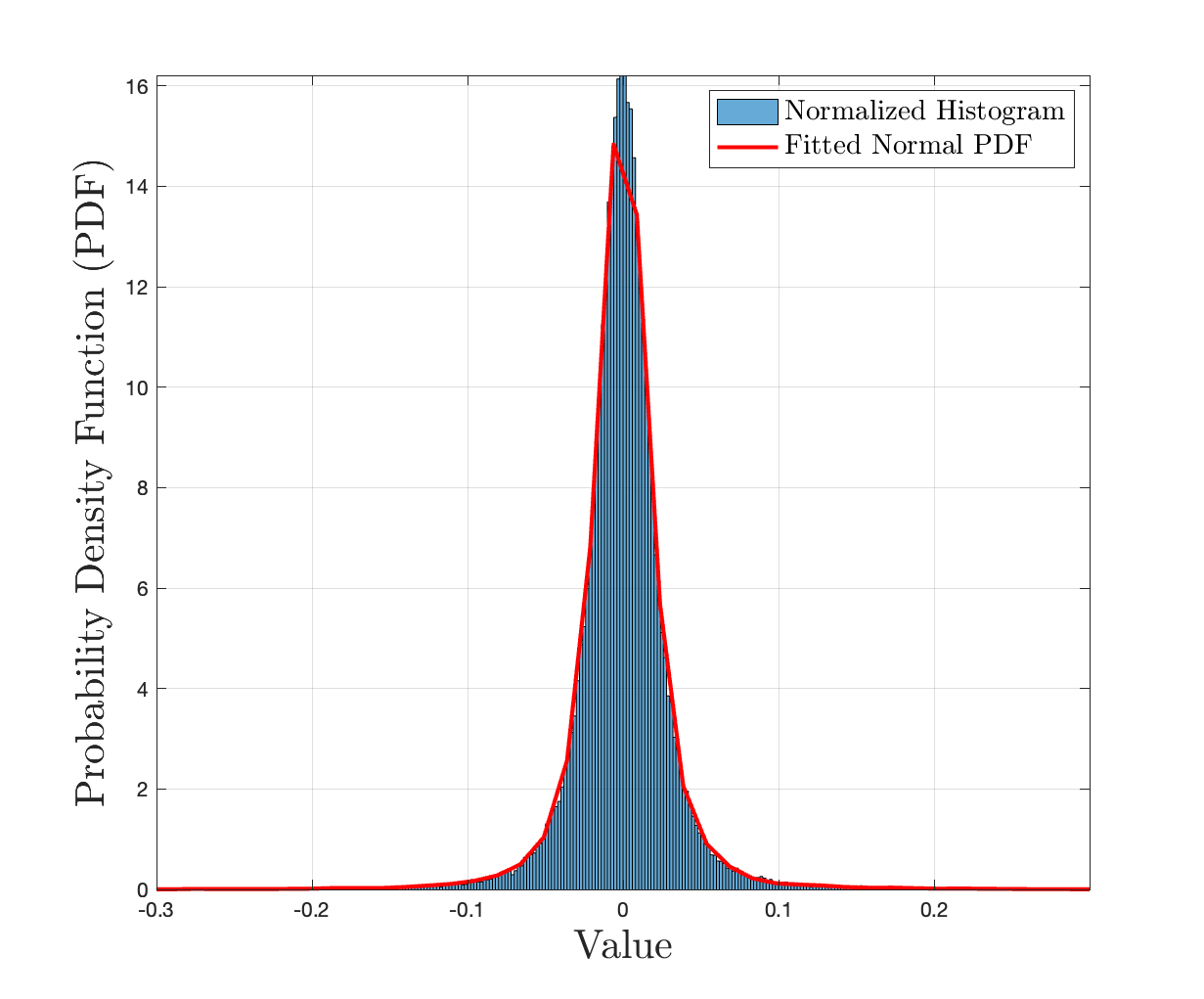}
        \caption{Normalized histogram of $\boldsymbol{c} - \boldsymbol{G}\boldsymbol{p}$ overlaid with a fitted normal distribution with zero mean and standard deviation $\sigma = 0.03$. The setup assumes $M = 16$ ULA BS antennas serving $K = 4$ single-antenna UEs at an SNR of 25~dB. The channel follows the Rician fading model described in Section~
\ref{sec:numericalchannelModel}. }
 \label{fig:histogram} \vspace{-4mm}
\end{figure}
The EP algorithm comprises two stages. First, we calculate the Gaussian approximation of the posterior distribution. Second, we calculate the precoder estimates based on the Gaussian approximation and update the parameters to refine the Gaussian approximation. The EP iteratively refines the Gaussian approximation until the convergence criterion is met.


\begin{itemize}
    \item \textbf{Gaussian approximation:} 
The EP framework is used to iteratively approximate $ p(\vect{p}|\vect{c}) $ in \eqref{eq_EP:Posterior_ori}. This is done first by approximating the two following distributions. \cite{Jespedes-TCOM14}: 
\begin{enumerate}
    \item The likelihood distribution $\mathcal{N} \left( \vect{c}: \vect{G} \vect{p} , {\sigma}^2 \vect{I} \right)$ by $ \mathcal{N}  \left( \vect{p}: {\vect{G}}^\dagger \vect{c} , \left( \vect{G}^{\Ttran} (\hat{\sigma}^2 \vect{I})^{-1} \vect{G} \right)^{-1} \right)$, which can be viewed as a least-squares transformation of the likelihood distribution.
    \item The a priori distribution $p(\vect{p})$  by an exponential family distribution, $\chi^{(t)}(\vect{p})\propto \mathcal{N} \left(\vect{p}: (\boldsymbol{\lambda}^{(t)})^{-1} \boldsymbol{\gamma}^{(t)}, (\boldsymbol{\lambda}^{(t)})^{-1}  \right)$, where $\boldsymbol{\lambda}^{(t)} $ is a diagonal matrix $K \times K$ with diagonal elements $\lambda_{m}^{(t)}>0 $  and $ \boldsymbol{\gamma}^{(t)} = [\gamma_{1}^{(t)}, \dots, \gamma_{M}^{(t)}]^{\Ttran}$. Both $\lambda_{m}^{(t)}$ and $\gamma_m^{(t)}$ are updatable parameters with $t$ denoting the iteration index. The parameters are used to enable the Gaussian approximation of $p(\vect{p}|\vect{c})$ and are initialized as $\lambda_{m}^{(0)}=1$ and $\gamma_m^{(0)}=0$. 
\end{enumerate}
Based on these approximations, the EP algorithm computes a tractable Gaussian distribution $ p^{(t)}(\vect{p}|\vect{c})$, which results from relaxing the discrete prior and replacing it with a Gaussian prior, parametrized by $\left(\boldsymbol{\gamma}^{(t)},\boldsymbol{\lambda}^{(t)}\right)$, expressed as
\begin{flalign}\label{eq_EP:Post_approx}
p^{(t)}(\vect{p}|\vect{c}) = & \mathcal{N}  \left( \vect{p}: {\vect{G}}^\dagger \vect{c}, \left( \vect{G}^{\Ttran} (\hat{\sigma}^2_{(t)} \vect{I})^{-1} \vect{G} \right)^{-1} \right)  \notag \\
&\cdot \mathcal{N}  \left(\vect{p}: (\boldsymbol{\lambda}^{(t)})^{-1} \boldsymbol{\gamma}^{(t)}, (\boldsymbol{\lambda}^{(t)})^{-1}  \right)\notag \\
\propto &\mathcal{N}  \left( \vect{p}:\boldsymbol{\mu}^{(t)}, \boldsymbol{\Sigma}^{(t)} \right),
\end{flalign}
where $\hat{\sigma}^2_{(t)}$ is the approximation of the  variance ${\sigma}^2$ in iteration $t$, and it is initialized as $\hat{\sigma}^2_{(0)}=1$. By applying the Gaussian product property,\footnote{The product of two Gaussian results in another Gaussian as given in \cite[App.~A.1]{Rasmussen-BOOK}, $\mathcal{N}_{\mathbb{C}}(\vect{p}:\vect{a},\vect{A}) \cdot \mathcal{N}_{\mathbb{C}}(\vect{p}:\vect{b},\vect{B})  \propto \mathcal{N}_{\mathbb{C}} (\vect{p}:(\vect{A}^{-1}+\vect{B}^{-1})^{-1}(\vect{A}^{-1} \vect{a} + \vect{B}^{-1} \vect{b}),(\vect{A}^{-1}+\vect{B}^{-1})^{-1}$.} we compute the product of two Gaussians in \eqref{eq_EP:Post_approx} and obtain the covariance matrix and mean vector of $p^{(t)}(\vect{p}|\vect{c}) $ as
\begin{subequations} \label{eA1_a0102}
            \begin{align}
&\boldsymbol{\Sigma}^{(t)} =  {\left(\vect{G}^{\Ttran} (\hat{\sigma}^2_{(t)} \vect{I})^{-1} \vect{G}+ \boldsymbol{\lambda}^{(t)} \right)}^{-1} ,\label{eA1_a01}\\
& \boldsymbol{\mu}^{(t)} =\boldsymbol{\Sigma}^{(t)} {\left(\vect{G}^{\Ttran}  (\hat{\sigma}^2_{(t)} \vect{I})^{-1} \vect{c} + \boldsymbol{\gamma}^{(t)}\right)}. \label{eA1_a02}
            \end{align}
        \end{subequations}
In what follows, we describe how the parameters $\vect{\gamma}^{(t)}$ and $\vect{\lambda}^{(t)}$  are iteratively computed to update the Gaussian approximation.
\item {\textbf{Precoder estimates and parameters update}:} 

We compute the cavity distribution in the EP algorithm \cite{Jespedes-TCOM14}. \footnote{The term `cavity' refers to the removal of the contribution of a specific factor (self-related information) from the Gaussian approximation distribution.} Given 
$\chi^{(t)}({p}_m)\propto \mathcal{N} \left({p}_m: {\gamma_m}^{(t)}/{(\lambda_m)}^{(t)} , (1/{(\lambda_m)}^{(t)})  \right)$ and the $m$-th marginal of $p^{(t)}(\vect{p}|\vect{c}) $, i.e.,  $\mathcal{N}  \left( {p}_m:{\mu}^{(t)}_m, {\Sigma}^{(t)}_{m,m} \right)$, the cavity marginal is computed as
 \begin{flalign}\label{eA1_a0304raw}
u^{(t)\backslash{m}}  &=\frac{\mathcal{N}  \left( {p}_m:{\mu}^{(t)}_m, {\Sigma}^{(t)}_{m,m} \right)}  {\mathcal{N} \left({p}_m:({\gamma}_m^{(t)}/{\lambda}_m^{(t)}), (1/{\lambda}_m^{(t)}) \right)},\notag\\
& \propto \mathcal{N}  \left( {p}_m: {p}^{(t)}_{{\rm obs},m}, {v}^{(t)}_{{\rm obs},m} \right),
 \end{flalign}
 where
 \begin{subequations} \label{eA1_a0304}
            \begin{align}
&v_{{\rm obs},m}^{(t)} =  \frac{\Sigma_{m}^{(t)} }{1- \Sigma_{m}^{(t)}  \lambda_{m}^{(t)}},  \label{eA1_a03}\\
&p_{{\rm obs},m}^{(t)}  = v_{{\rm obs},m}^{(t)}  {\left(\frac{\mu_{m}^{(t)}}{\Sigma_{m}^{(t)}}-\gamma_{m}^{(t)}\right)}.  \label{eA1_a04} 
            \end{align}
        \end{subequations}
Here, $\mu_m^{(t)}$ is the $m$-th element of vector $\boldsymbol{\mu}^{(t)}$ and $\Sigma_m^{(t)}$ is the $m$-th diagonal element of matrix $\boldsymbol{\Sigma}^{(t)}$. We can now compute a new posterior distribution defined over the discrete constellation space, formed by multiplying the cavity distribution with the discrete prior induced by the constellation. The first two moments of this posterior are then given by:
\begin{subequations}\label{eA1_b0102}
\begin{equation}\label{eA1_b01}
\hat{{p}}^{(t)}_m =  \sum_{l \in \mathcal{L}}  l \times u^{(t)\backslash{m}}(p_m=l),
\end{equation}
\begin{equation}\label{eA1_b02}
v_m^{(t)} =  \sum_{l \in \mathcal{L}} \left( x_m - \hat{x}_m^{(t)} \right)^2 \times u^{(t)\backslash{m}}(p_m=l).
\end{equation}
\end {subequations}
Using the first two moments of the new posterior distribution, we refine the Gaussian prior distribution i.e.,  $\chi^{(t)}(\vect{p})$ in \eqref{eq_EP:Post_approx} \cite{Jespedes-TCOM14}:
 \begin{flalign} \label{eq:EP_inference_recon}
\chi^{(t+1)}(\vect{p}) &= \frac{\mathcal{N}  \left( \vect{p}:  \hat{\vect{p}}^{(t)}, \vect{V}^{(t)} \right)}{\mathcal{N} \left( \vect{p}: \vect{p}^{(t)}_{{\rm obs}},\vect{V}^{(t)}_{{\rm obs}}\right) } \notag \\
 &= \mathcal{N}  \left(  \vect{p}: (\boldsymbol{\lambda}^{(t+1)})^{-1} \boldsymbol{\gamma}^{(t+1)},  (\boldsymbol{\lambda}^{(t+1)})^{-1}   \right),
\end{flalign}  
where $\vect{V}^{(t)}$ is a diagonal matrix with its $m$-th diagonal element is $v^{(t)}_m$, $\vect{V}^{(t)}_{{\rm obs}}$ is a diagonal matrix with its $m$-th diagonal element is $ v_{{\rm obs},m}^{(t)}$,   $\vect{p}^{(t)}_{{\rm obs}} = [{p}^{(t)}_{{\rm obs,1}}, \cdots, \vect{p}^{(t)}_{{\rm obs},M}]$, and
\begin{subequations} \label{eA1_b0304}
            \begin{align}
&\boldsymbol{\lambda}^{(t+1)} = (\vect{V}^{(t)})^{-1} -  (\vect{V}^{(t)}_{{\rm obs}})^{-1},  \label{eA1_b03}\\
&\boldsymbol{\gamma}^{(t+1)} =  (\vect{V}^{(t)})^{-1} \hat{\vect{p}}^{(t)} - (\vect{V}^{(t)}_{{\rm obs}})^{-1}   \vect{p}^{(t)}_{{\rm obs}}. \label{eA1_b04}
            \end{align}
\end{subequations}
 The Gaussian posterior distribution in the next iteration ${p}^{(t+1)}(\vect{p}|\vect{c})$ is calculated on the basis of the refined prior distribution, parameterized by $\boldsymbol{\lambda}^{(t+1)}$ and $\boldsymbol{\gamma}^{(t+1)}$. 
We smoothen the update of $(\boldsymbol{\lambda}^{(t+1)},\boldsymbol{\gamma}^{(t+1)})$ by using a linear combination with the former values:
\begin{subequations} \label{eq:damping}
\begin{align}
   \boldsymbol{\lambda}^{(t+1)} &\leftarrow  (1-\eta)\boldsymbol{\lambda}^{(t+1)}+\eta \boldsymbol{\lambda}^{(t)}, \label{damping_lambda} \\
   \boldsymbol{\gamma}^{(t+1)} &\leftarrow  (1-\eta)\boldsymbol{\gamma}^{(t+1)}+\eta\boldsymbol{\gamma}^{(t)}\label{damping_gamma},
\end{align}
\end{subequations}
where $\eta\in[0, 1]$ is a predetermined weighting coefficient. The smoothened (damped) update is introduced to improve the numerical stability and convergence of the expectation propagation algorithm, as full EP updates are not guaranteed to be stable and may lead to oscillations or divergence. Damping is a standard technique in EP that forms a convex combination of the previous and updated estimates, balancing convergence speed and robustness.

Before proceeding to the next iteration, we estimate the variance $\sigma^2$ based on the estimated precoder vector in \eqref{eA1_b01}. Specifically, the precoder error variance is 
\begin{equation}\label{eq_var_estim}
   \hat{\sigma}^2_{(t+1)} = \frac{\|\vect{c} - \vect{G} \hat{\vect{p}}^{(t)}\|_2^2 + \epsilon^{(t)} }{M},  
\end{equation}
where \(\epsilon^{(t)} = (\hat{\sigma}^2_{(t)} - \hat{\sigma}^2_{(t-1)})^2 \) is the variance precision error, initialized to 0. Although we could use the covariance matrix \((\vect{V}^{(t)})\) to calculate the estimation variance \(\epsilon^{(t)}\), as noted in \cite{hayakawa2023noise}, the very small values of its elements do not accurately represent precision error estimates. We address this by employing a variance precision error based on the difference between variance estimations across successive iterations. 

\end{itemize}

The updated parameters, along with the estimated precoder error variance, are then used to refine the Gaussian approximation in \eqref{eq_EP:Post_approx}. The iterations are subsequently continued.
If  the maximum number of iterations $T$ has been reached, the EP iteration is completed and hard estimates are derived from the soft precoder estimates $\hat{\vect{p}}^{(T)} = [\hat{{p}}^{(T)}_1,\cdots,\hat{{p}}^{(T)}_M]^{\Ttran}$ by mapping them to the closest constellation points in $\mathcal{L}$ based on their Euclidean distances. The estimates are then converted into the complex domain to obtain the detected symbols.
We show the complete pseudocode in Algorithm \ref{A1}. 

\begin{algorithm}[t!]
\caption{EP Algorithm for Solving \eqref{eq:sdreduced notation}}
\label{A1}
\begin{algorithmic}[1]
\State \textbf{Input:} $\boldsymbol{\gamma}^{(0)} = \boldsymbol{0}$, $\boldsymbol{\lambda}^{(0)} = \vect{I}$, $\eta$, $\hat{\sigma}^2_{(0)} = 1$, $T_{\rm max}$

\FOR {$t = 1$ to $T_{\rm max}$}
    \State Compute $\boldsymbol{\Sigma}^{(t)}$ and $\boldsymbol{\mu}^{(t)}$ using \eqref{eA1_a0102}
    \State Compute $\vect{V}_{\rm obs}^{(t)}$ and $\vect{p}_{\rm obs}^{(t)}$ using \eqref{eA1_a0304}    
    \State Compute $\vect{V}^{(t)}$ and $\hat{\vect{p}}^{(t)}$ using \eqref{eA1_b0102}
    \State Update $\boldsymbol{\lambda}^{(t+1)}$ using \eqref{eA1_b03} and apply damping via \eqref{damping_lambda}
    \State Update $\boldsymbol{\gamma}^{(t+1)}$ using \eqref{eA1_b04} and apply damping via \eqref{damping_gamma}
    \State Update $\hat{\sigma}^2_{(t+1)}$ using \eqref{eq_var_estim}
\ENDFOR

\State \textbf{Output:} Hard estimates of $\hat{\vect{p}}^{(T_{\rm max})}$
\end{algorithmic}
\end{algorithm}

 \section{Alternative Precoding Schemes} \label{sec:heuristic}

Instead of optimally solving problem $\mathbb{P}_7$, several alternative approaches can be considered as baseline solutions for the iterative WMMSE problem $\mathbb{P}_3$. In the following, we introduce some of these methods.

\subsection{Quantization-Unaware Precoding}

The conventional WMMSE-based algorithms in \cite{christensen2008weighted,zhao2022rethinking} (among others) compute a precoding matrix from $\mathbb{C}^{M\times K}$ that maximizes the sum rate based on infinite resolution and the available CSI. Hence, the naive baseline approach would be to compute such a precoding matrix $\boldsymbol{P}^{\mathrm{unquantized}} \in \mathbb{C}^{M\times K}$ and then quantize each entry using $\mathcal{Q}(\cdot): \mathbb{C}^{M \times K} \to \mathcal{P}^{M \times K}$ so that the result can be sent over the fronthaul. In this case, the BBU follows Algorithm~\ref{Alg:WMMSE}  but solves $\mathbb{P}_3$ in the domain $\mathbb{C}$ instead of $\mathcal{P}$, and $\mathbb{P}_3$ becomes a continuous convex optimization problem that is solvable using any general-purpose convex solver.
We will refer to this quantization-unaware precoding approach as the \emph{Unaware} precoding. After the WMMSE algorithm has converged, the final precoding matrix is obtained by quantizing each entry as $\boldsymbol{P} = \mathcal{Q}(\boldsymbol{P}^{\mathrm{unquantized}})$.

\subsection{Quantization-Half-Aware Precoding}

Although the proposed SD-based precoding gives the near-optimal (due to the bisection method) solution for problem $\mathbb{P}_3$, EP-based precoding is not optimal but is a much faster approach. It is worth exploring the sum rate with an alternative that searches for a precoding matrix in $\mathbb{C}^{M \times K}$ for $N-1$ iterations and then uses the SD-based algorithm only for the final iteration to compare with the EP-based approach. In this case, we will first identify UE gains $\beta_k$ and weights $d_k$ that are suitable for sum-rate maximization with infinite-resolution precoding, and then compute the corresponding optimized quantization-aware precoding, which we call \emph{Half-aware}. 
\subsection{Heuristic Precoding} 
Another benchmark could be a heuristic precoding scheme that attempts to improve on the \textit{Unaware} precoding. After we quantize the precoding matrix $\boldsymbol{P}^{\mathrm{unquantized}}$ to obtain $\boldsymbol{P} $, we refine the elements sequentially, where the ordering should be done properly. We consider the
second closest quantization levels in both the real and imaginary dimensions according to Euclidean distance \cite{khorsandmanesh2023optimized}. This search gives us three alternative ways of quantizing each element in the precoding matrix $\boldsymbol{P}^{\mathrm{unquantized}}$. We call this method the \emph{Heuristic}. We propose to start by updating the column of the quantized precoding matrix corresponding to the UE $k$ with the highest \textit{generated interference} $\mathrm{GI}_k = \sum_{{i=1} , i \ne k}^K |[{\boldsymbol{H}} {\hat{\boldsymbol{P}}}]_{i,k}|^2$, where $\hat{\boldsymbol{P}} = \alpha \boldsymbol{P} = \alpha \mathcal{Q} (\boldsymbol{P}^{\mathrm{unquantized}})$ since this might improve the performance the most.\footnote{We have noticed experimentally that this leads to the largest improvement in sum rate at high SNR.} 
We evaluate the sum rate
\begin{equation}  \sum_{k=1}^K 
\log_2 \left(1+\frac{\big|[{\boldsymbol{H}} \hat{\boldsymbol{P}}]_{k,k}\big|^2}{\sum_{i=1, i  \ne k}^{K}\big|[{\boldsymbol{H}} \hat{\boldsymbol{P}}
]_{k,i}\big|^2 +\bar{N_0}} \right),
\label{eq:sumrate}
\end{equation}
for the four different ${\boldsymbol{P}}$ options obtained with $p_{k,m} \in \{ \text{four nearest points to } {P}^{\mathrm{unquantized}}_{k,m} \hspace{1mm} \text{in} \hspace{1mm} \mathcal{P} \} $ while all other elements are fixed. 
We then replace the corresponding element in $\boldsymbol{P}$ with the option that achieves the highest sum rate. The rest of the UEs are ordered based on decreasing generated interference, and the precoding elements are updated element-wise similarly.

\section{Numerical Results}\label{secnumerical}

This section compares the sum rates achieved by the aforementioned precoding approaches as a function of the SNR. The sum rate is calculated using Monte Carlo simulations for the case of Gaussian signaling.

\subsection{Channel Model} \label{sec:numericalchannelModel}
The following system configuration and parameters are assumed for the MU-MIMO setup unless specified otherwise. The AAS at the BS is equipped with $M = 16$ antennas. The number of quantization levels is $L=8$, and the number of UEs is $K = 4$.
We consider Rician fading channels,  composed of a line-of-sight (LoS) path component and a non-line-of-sight (NLoS) path component, computed as
\begin{equation}
    \boldsymbol{h}_k = \sqrt{\rho_k} \Big( \sqrt{\frac{\kappa}{\kappa+1}}\boldsymbol{h}_k^{\text{LoS}}+\sqrt{\frac{1}{\kappa +1}}\boldsymbol{h}_k^{\text{NLoS}}\Big),
\end{equation}
where $\kappa$ is the Rician factor,  set as $\kappa = 10$, while $\boldsymbol{h}_k^{\text{LoS}}\in \mathbb{C}^{M}$ and $\boldsymbol{h}_k^{\text{NLoS}} \in \mathbb{C}^{M}$ are the LoS and NLoS components, for $k$-th UE respectively. We assume that the AAS is a ULA with half-wavelength spacing; thus, the LoS channel matrix can be modeled as 
$\vect{h}_k^{\mathrm{LoS}} = \vect{a}(\Omega_k) ,$ where
\begin{align}
\vect{a}(\Omega_k) &= \left[1, e^{j\psi_{k}}, \ldots, e^{j(M-1)\psi_{k}} \right]^{\Ttran},
\end{align}
and $\vect{a}(\cdot)$ denotes the array response vectors of the AAS. The relative phase shift is defined as
$\psi_{k} = \pi \sin(\Omega_k)$. The angles $\Omega_k$ for the $k$-th UE correspond to the angle of departure from the AAS towards UE $k$. In the simulations, we generate $\Omega_k$ randomly  with uniform distribution as $\Omega_k \sim \mathcal{U}[-\pi/3, \pi/3]$.

The NLoS component $\vect{h}_k^{\mathrm{NLoS}}$ is modeled as uncorrelated Rayleigh fading with $ \mathcal{CN}(0,1)$-entries. We assume that the UEs are uniformly distributed such that the distance $d_k$ from UE $k$ to the AAS $d_k  \sim \mathcal{U}[10\,\mathrm{m}, 200\,\mathrm{m}]$. Furthermore, the path loss $\rho_k$ between the AAS and UE $k$ is modeled according to the 3GPP model outlined in \cite[Table 7.4.1-1]{3gpp2018study}, tailored for an urban microcell (UMi) environment and disregarding shadow fading as:

\begin{equation}
\rho_k = -37.5 - 22 \log_{10} \left( \frac{d_k}{1~\mathrm{m}} \right) \quad [\mathrm{dB}],
\end{equation}
where the carrier frequency is $3$\,GHz.
In the following, the SNR is defined as $\mathrm{SNR}=\frac{q }{N_0}$.

\subsection{Convergence}

\begin{figure}[!t]
        \centering
      \includegraphics[scale=0.2]{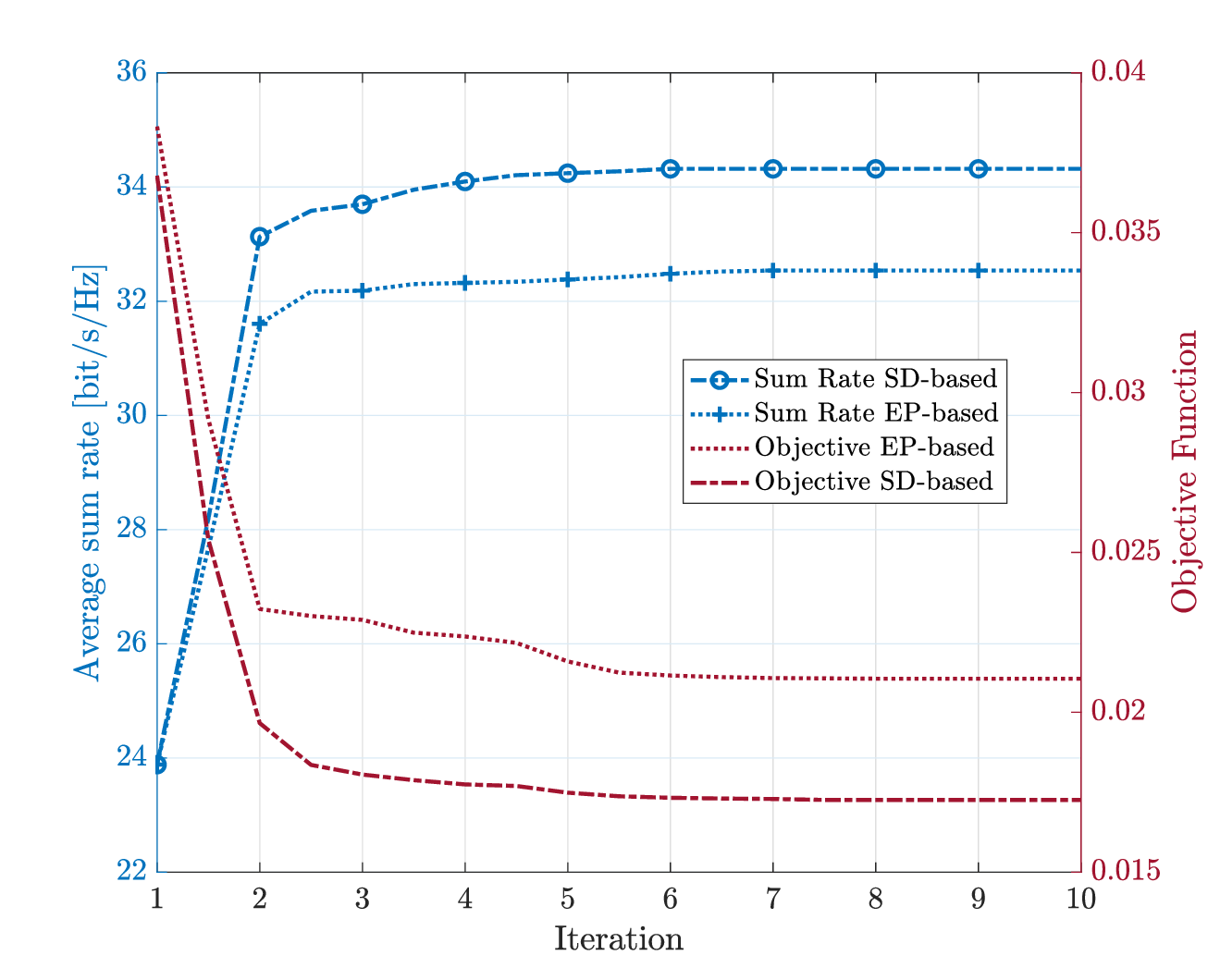}
        \caption{The average sum rate and objective function \eqref{eq:wmmse} evolution when running the proposed WMMSE algorithm using the proposed SD-based and EP-based methods.}
 \label{fig:convergence} \vspace{-3mm}
\end{figure}

Fig.~\ref{fig:convergence} presents the convergence behavior of the proposed WMMSE Algorithm~\ref{Alg:WMMSE} for the cases when $\mathbb{P}_3$ is solved using the proposed SD-based and EP-based methods. We consider $\mathrm{SNR}=20$ dB. The objective function shown on the right $y$-axis corresponds to \eqref{eq:wmmse}. The proposed SD-based and EP-based algorithms reach a stationary point after $6$ and  $7$ iterations, respectively. The objective function converges to a lower value when using the SD-based approach, indicating that the SD-based design finds a better precoding solution. As a result, the SD-based scheme also achieves a higher sum rate (shown on the left $y$-axis). The EP-based design incurs a slight performance degradation relative to its SD-based counterpart. The increase in sum rate from initialization to convergence illustrates the improvement over the sum-MSE method in \cite{khorsandmanesh2023optimized}. SD- and EP-based methods converge at comparable speeds.

 \subsection{Computational Complexity and Run-time}

\begin{table}[t!]
\centering
\caption{Complexity of the WMMSE algorithm for the different precoding schemes, where $N$ is the maximum number of WMMSE iterations, $M$ is the number of BS antennas, and $K$ is the number of UEs.}
\resizebox{0.8\columnwidth}{!}{
\begin{tabular}{|>{\centering\arraybackslash}m{2.5cm}|>{\centering\arraybackslash}m{5cm}|}
\hline
\textbf{Precoding Method} & \textbf{Computational Complexity} \\
\hline
\textbf{SD-based} & $O(NKL^{2\gamma M})$ \\
\hline
\textbf{EP-based} & $O(NK^3)$ \\
\hline
\textbf{Half-aware} & $O(NKM + KL^{2\gamma M})$ \\
\hline
\textbf{Heuristic} & $O(NMK)$ \\
\hline
\textbf{Unaware} & $O(NMK)$  \\
\hline
\end{tabular}
}
\vspace{-4mm}
\label{tab:complexity}
\end{table}
\begin{figure}[!t]
        \centering      \includegraphics[scale=0.23]{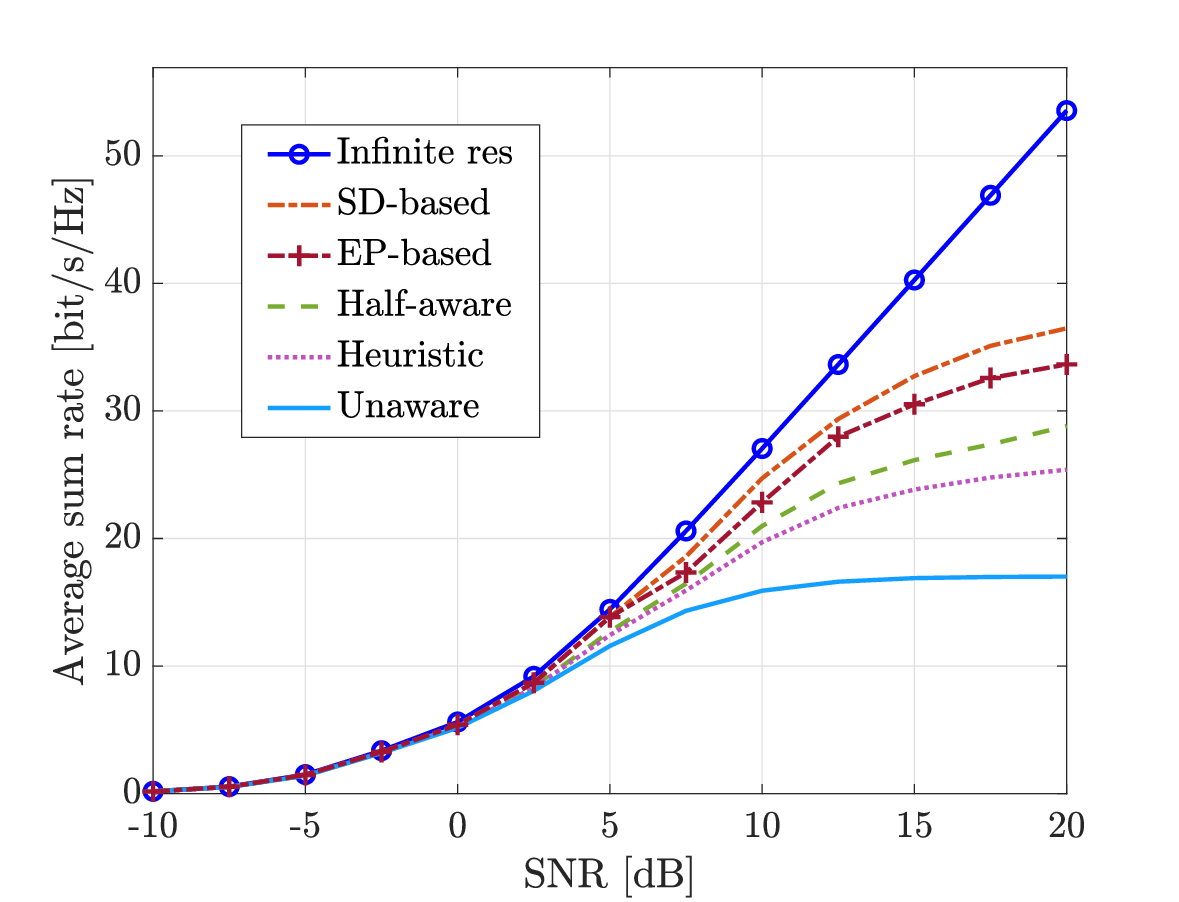}
        \caption{The average sum rate versus the SNR for different precoding schemes, assuming $M=4 \times 4$, $K=4$, and $L=8$.}
 \label{fig:sumrate} \vspace{-6mm}
\end{figure}

Before comparing the sum rate results obtained with the different precoding methods, it is worth studying their complexity, because there is a fundamental tradeoff between communication performance and complexity. 

The complexity of the WMMSE algorithm mainly originates from optimizing the precoding vectors $ \vect{p}_i$ for $k, i = 1,\ldots, K$, which depends on the algorithm that is used to solve the subproblem $\mathbb{P}_7$. Utilizing the SD algorithm has $O(K L^{2\gamma M})$ for some $0 < \gamma \leq 1$. Note that the value of $\gamma$ depends on problem-specific parameters, such as the channel statistics, and cannot be trivially determined  \cite{jalden2005complexity}. The complexity of EP depends on the size of the matrix $\mathbf{G}$, which is $K \times K$ in our case, and as we solve the problem for $K$ UEs, it becomes cubic \cite{wang2020expectation}. The heuristic method solves the same problem as \emph{Unaware}, but repeated four times. $N$ denotes the maximum number of iterations of the WMMSE algorithm. Table~\ref{tab:complexity} summarizes the complexity orders of the WMMSE algorithm with different kinds of precoding schemes.

Fig.~\ref{fig:sumrate} depicts the average sum rate as a function of the SNR for the different precoding schemes. To consider a more realistic scenario to generate this figure, we model the NLoS channel as $\boldsymbol{h}_k^{\text{NLoS}} \sim \mathcal{CN}(\boldsymbol{0}_M,\boldsymbol{R}_{k})$, which is spatially correlated Rayleigh fading. The spatial correlation matrix ${\boldsymbol{R}_{k}}\in  \mathbb{C}^{M \times M} $ is generated following the local scattering model from [34,Ch.~2]. The locations of the UEs are randomly generated with the same elevation angle $\theta =0$ but uniformly distributed azimuth angles $\phi$, seen from the UPA. The top curve, \emph{infinite res}, considers the ideal case without quantization and outperforms all the quantized precoding schemes since the rate increases linearly (in dB scale) at high SNR. In all quantized precoding schemes, the sum rate converges to specific limits at high SNR since the interference cannot be canceled entirely due to the limited precoding resolution; i.e., the system is interference-limited at high SNR. The gap between \emph{infinite res} and our novel \emph{SD-based} precoding is remarkably smaller than the gap between \emph{Unaware} and \emph{infinite res}, where we quantized the precoding matrix used for \emph{infinite res}. Our algorithm provides twice the rate at high SNR. The lower-complexity \emph{EP-based}
also outperforms other approaches \emph{Half-aware}, \emph{Heuristic} and  \emph{Unaware}, but there is a slight gap to \emph{SD-based}. However, there would be a trade-off between performance in terms of sum rate and complexity to choose between \emph{SD-based} and \emph{EP-based}. The substantial gap between \emph{Half-aware} and  \emph{SD-based} demonstrates the importance of computing quantization-aware weights in the WMMSE algorithm instead of relying on those obtained with infinite resolution. 
The proposed \emph{Heuristic} precoding reaches nearly the same performance as \emph{Half-aware} and, thus,  performs vastly better than \emph{Unaware} precoding. Note that the complexity of \emph{EP-based} is cubic with $K$, while  \emph{Heuristic} and  \emph{Unaware}  are polynomial with $M$ and $K$. Therefore, all are implementable in e.g., mMIMO systems.

All reported run-times are obtained with $L=8$ quantization levels and averaged over $500$ independent channel realizations. The \emph{SD-based} method exhibits the highest computational cost, with an average run-time of $432.92$ seconds, due to the complexity of the discrete search. In contrast, the \emph{EP-based} and \emph{Unaware} schemes require only $10.09$ and $9.32$ seconds, respectively, making \emph{EP-based} approximately $40$ times faster than \emph{SD-based}. The comparable run-time of \emph{EP-based} and \emph{Unaware} methods indicates that \emph{EP-based} achieves a favorable complexity–performance trade-off and is more suitable for real-time implementation.

\begin{figure}[!t]
        \centering
      \includegraphics[scale=0.25]{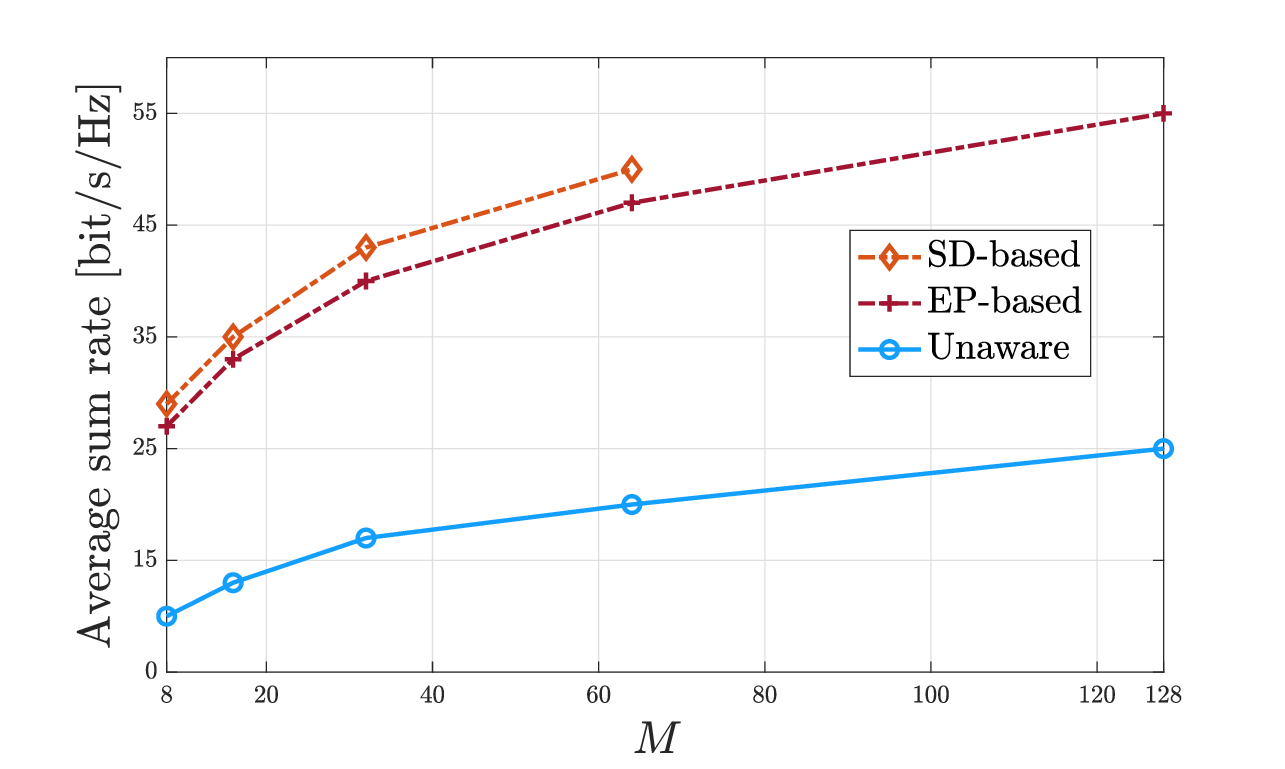}
        \caption{The average sum rate versus the BS antennas $M$ for various precoding schemes, assuming $\text{SNR} = 20$ dB, $K=4$ UEs and $L=8$ quantization levels.}
 \label{fig:BSantennas} \vspace{-5mm}
\end{figure}
Fig.~\ref{fig:BSantennas} shows the average sum rate as a function of the number of BS antennas $M$. By increasing $M$, the average sum rate of three precoding schemes \emph{SD-based}, \emph{EP-based}, and \emph{Unaware} demonstrates an upward trend thanks to increased beamforming gains and more dimensions for interference suppression. Notably, the performance gap between the schemes also widens with larger antenna arrays. While the \emph{SD-based} scheme consistently achieves the highest performance, it becomes computationally prohibitive for
$M > 64$. Compared to the \emph{Unaware} approach, the gap between the \emph{SD-based} and \emph{EP-based} precodings remains relatively small, confirming that the \emph{EP-based} method offers a favorable trade-off between complexity and performance.

\begin{figure}[!t]
        \centering
      \includegraphics[scale=0.25]{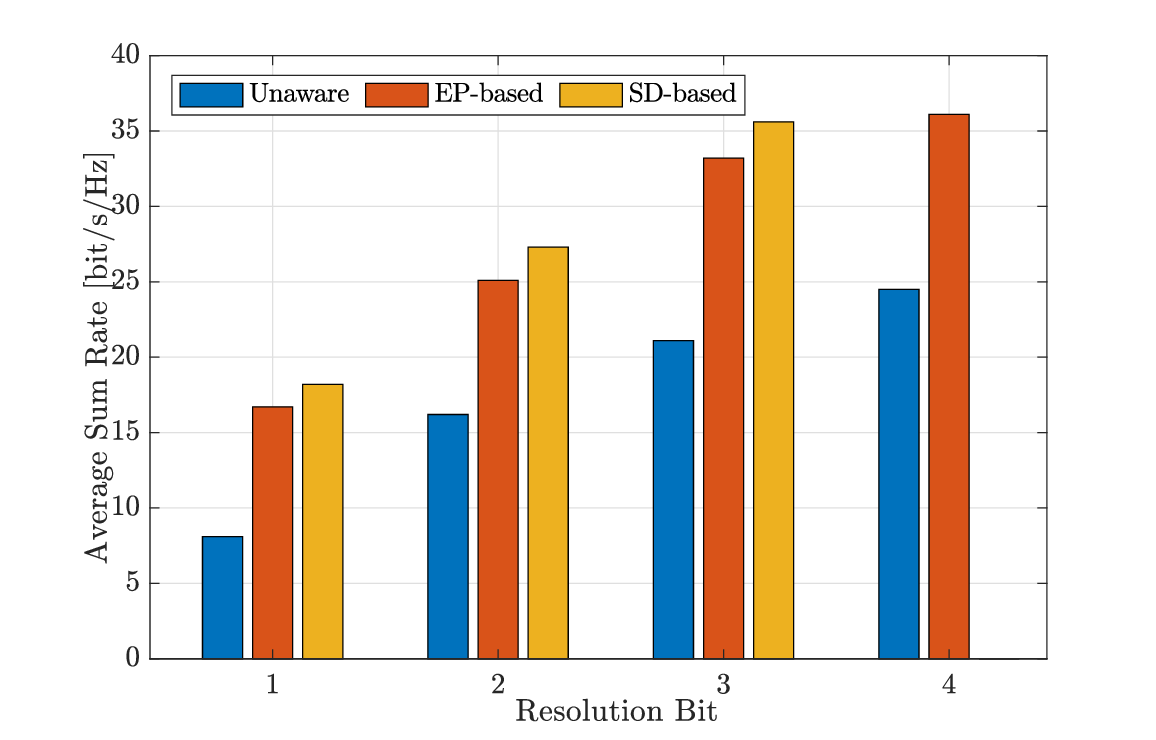}
        \caption{The average sum rate versus the resolution bit. We assume that the $\text{SNR} = 20$ dB, $M=16$, and $K=4$.}
 \label{fig:bit} 
 \vspace{-4mm}
\end{figure}
Fig.~\ref{fig:bit} shows the impact of fronthaul resolution on the average sum rate performance. As depicted, the average sum rate is improved with a higher bit resolution since it provides more flexibility in optimizing the precoder. The \emph{SD-based} method remains computationally feasible up to 3-bit resolution in the considered setup, beyond which its complexity becomes prohibitive. In contrast, the \emph{EP-based} approach maintains practical implementability across a broader range of resolutions. It consistently outperforms the \emph{Unaware} scheme, offering a favorable balance between performance and complexity. Hence, optimality comes at the cost of increased complexity, and we can reduce the complexity by sacrificing performance.

\subsection{Imperfect and Quantized CSI}\label{sec:numericalimperfect}

Fig.~\ref{fig:imperfect} investigates the impact of imperfect CSI and quantized estimated CSI, which is presented in Section~\ref{sec:estimation}. The figure shows the average sum rate as a function of the SNR for three cases: perfect CSI at the BBU (as assumed earlier in this section), imperfect CSI, which is obtained based on \eqref{eq:channelMLuplink}, and quantized imperfect CSI based on \eqref{eq:CQCSI}. We assume the same bit resolution for the CSI and precoding matrix and set $B_{\rm H} = 3$ (bits per complex entry). The BBU applies the same algorithms in all cases, while treating the CSI as perfect even if it is not. The average sum rate is still computed using the objective function in \eqref{eq:firstproblem}, assuming perfect CSI at the receiving UEs.  It can be observed that imperfect CSI leads to a reduction in the sum rate at low SNR and even more loss due to CSI quantization for the quantized estimated CSI case. However, the \emph{SD-based} precoding method still achieves a higher sum rate than the \emph{EP-based} alternatives. At high SNR, the difference in performance between the perfect and imperfect CSI cases becomes negligible because the quality of the CSI improves with the SNR. The quantized estimated CSI still has a gap compared to the other two cases. This shows that even though we assumed perfect CSI when designing our algorithms, the same methodology can be applied in scenarios with imperfect CSI and quantized ones. The impact of CSI imperfections is significant only at low SNR, whereas for quantized estimated CSI, the gap remains due to the quantization error of CSI.

\begin{figure}[!t]
        \centering
        \includegraphics[scale=0.19]{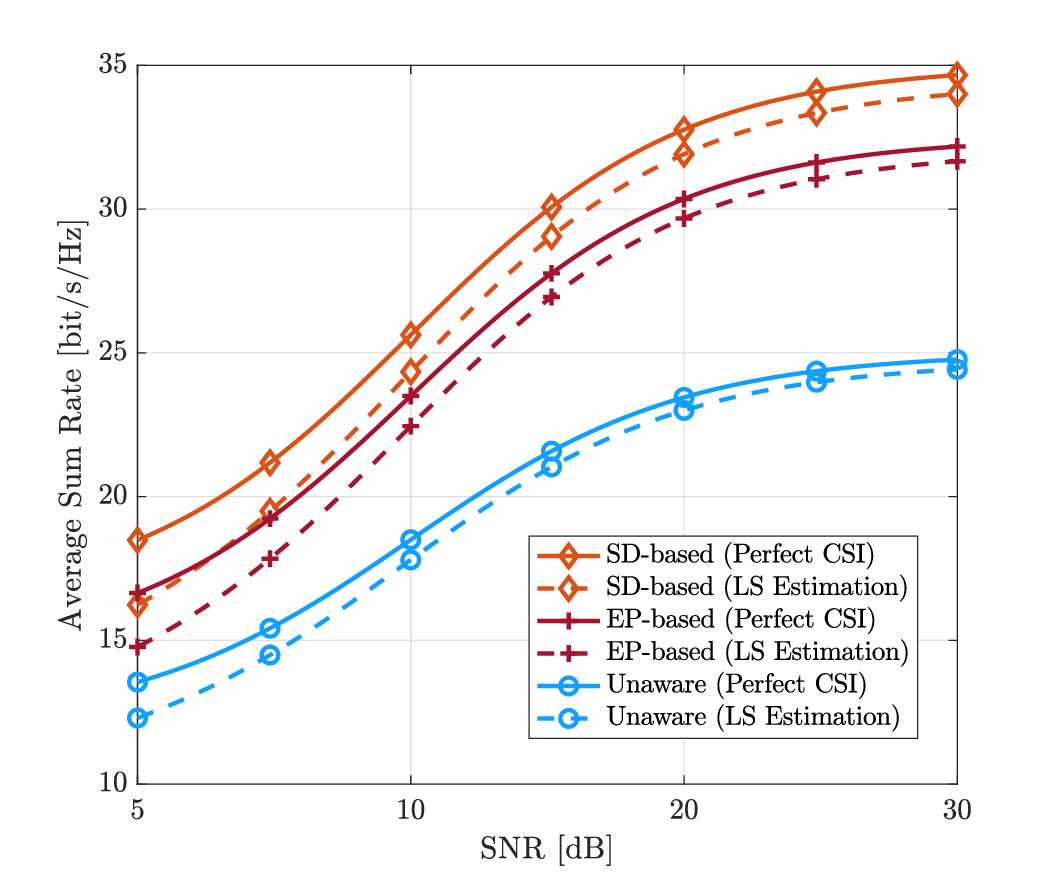}
        \caption{Average sum rate versus the  SNR with either perfect or imperfect CSI at the BBU, with
        $M=16$, $K=4$, and $L=8$.}
    \label{fig:imperfect}
        \vspace{-5mm}
\end{figure}

\subsection{Weighted Sum Rate Maximization}

The proposed Algorithm~\ref{Alg:WMMSE} can also be applied to weighted sum rate maximization, as discussed in Remark~\ref{rem:weighted}. In Fig.~\ref{fig:wightedsumrate}, we examine the impact of varying UE weights on the weighted sum rate. The simulation setup includes a BS equipped with $M=32$ ULA antennas serving  $K = 8$ UEs and quantization with 
$L=8$ levels. The average weighted sum rate is plotted against the SNR. We compare two weighting scenarios: (1) randomly assigned UE weights, where we uniformly select integers from the set $\{1, 2\}$ \textcolor{black}{ and normalize them such that the sum is equal to $k$}, and (2) uniform weights, where all UEs are assigned equal priority. The results demonstrate that the proposed algorithm can effectively prioritize UEs with higher weights. Specifically, in the case of unequal weights, the system achieves a higher weighted sum rate across all transmit power values compared to the equal-weight scenario, validating the algorithm's ability to leverage UE prioritization for enhanced resource allocation. Moreover, we evaluate \emph{SD-based}, \emph{EP-based}, and \emph{Unaware} precoding. The \emph{SD-based} and \emph{EP-based} schemes consistently achieve optimal and near-optimal performance, respectively. The marginal performance gap between \emph{SD-based} and \emph{EP-based} also suggests that the computationally simpler EP scheme may be a viable alternative in practical deployments where complexity is a concern.

\begin{figure}[t]
    \centering
    \includegraphics[scale= 0.25]{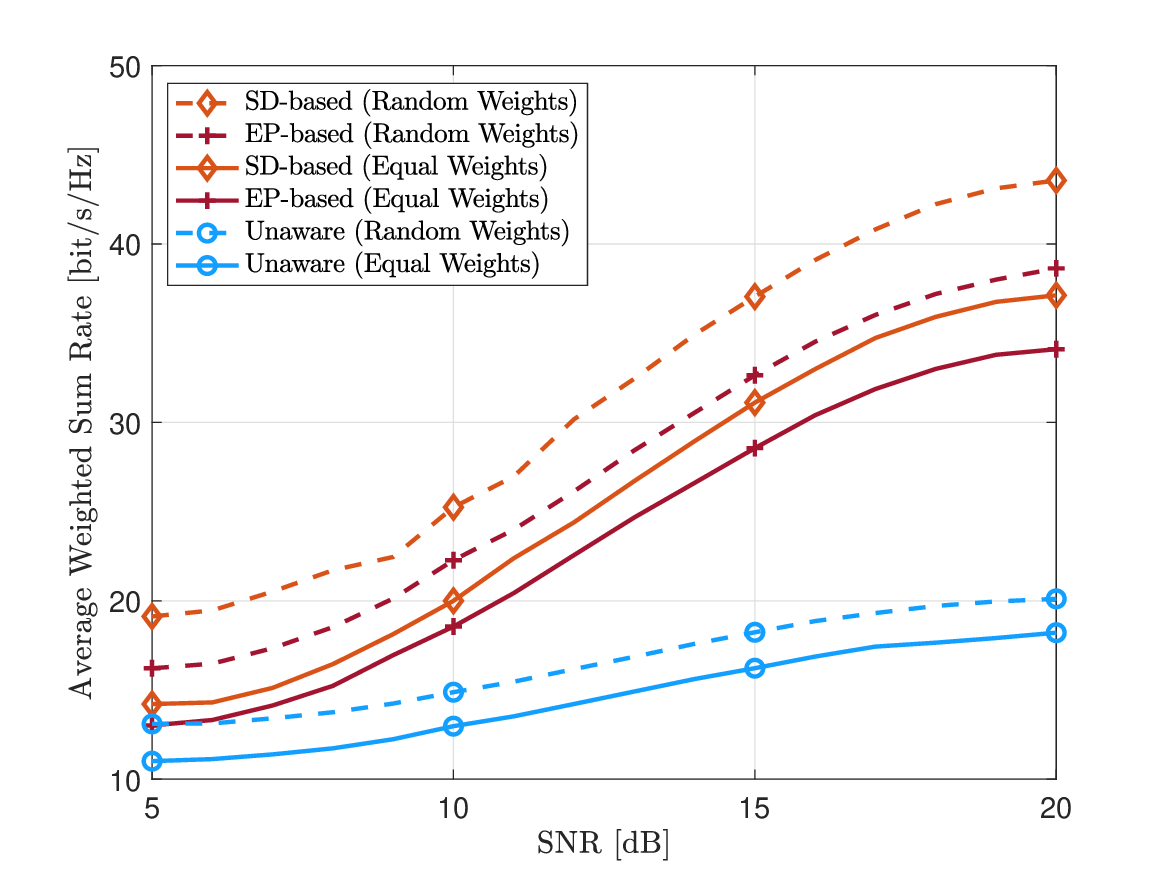}
    \caption{Weighted sum rate versus maximum downlink power with either equal weights or normalized random weights {1,2}.}
\vspace{-5mm}\label{fig:wightedsumrate}
\end{figure}

\section{Conclusions} \label{sec:concl}

A 5G site often consists of an AAS connected to a BBU via a digital fronthaul with limited capacity. Hence, the precoding matrix that is computed at the BBU must be quantized to finite precision. We have proposed a novel WMMSE-based framework for quantization-aware precoding that finds a local optimum to the sum rate maximization problem. Moreover, a reduced-complexity SD-based algorithm was proposed to find an approximate but tight solution.
We demonstrated that the sum rate can be doubled at high SNR by being quantization-aware when selecting the weights in the WMMSE formulation and the precoding matrix. Besides, an alternative quantization-aware EP-based precoding was proposed to outperform the quantization-unaware baseline with comparable complexity, making it feasible for larger setups. Finally, the framework can also be easily modified to solve  \emph{weighted} sum rate problems.

\bibliographystyle{IEEEtran}
\bibliography{IEEEabrv,references}

\end{document}